\documentclass[pra,aps,twocolumn,superscriptaddress,english]{revtex4-2}
\usepackage[T1]{fontenc}
\usepackage[utf8]{inputenc}
\usepackage{lipsum,babel}
\usepackage{mathtools}
\usepackage{amsthm}
\usepackage{hyperref}

\usepackage{amsmath,amssymb,mathrsfs}

\usepackage{graphicx}
\usepackage[matrix,frame,arrow]{xy}
\usepackage{amsmath}

\newcommand{\qw}[1][-1]{\ar @{-} [0,#1]}

\newcommand{\gate}[1]{*{\xy *+<.6em>{#1};p\save+LU;+RU **\dir{-}\restore\save+RU;+RD **\dir{-}\restore\save+RD;+LD **\dir{-}\restore\POS+LD;+LU **\dir{-}\endxy} \qw}

\newcommand{\measureD}[1]{*{\xy*+=+<.5em>{\vphantom{\rule{0em}{.1em}#1}}*\cir{r_l};p\save*!R{#1} \restore\save+UC;+UC-<.5em,0em>*!R{\hphantom{#1}}+L **\dir{-} \restore\save+DC;+DC-<.5em,0em>*!R{\hphantom{#1}}+L **\dir{-} \restore\POS+UC-<.5em,0em>*!R{\hphantom{#1}}+L;+DC-<.5em,0em>*!R{\hphantom{#1}}+L **\dir{-} \endxy} \qw}

\newcommand{\multimeasureD}[2]{*+<1em,.9em>{\hphantom{#2}}\save[0,0].[#1,0];p\save !C *{#2},p+LU+<0em,0em>;+RU+<-.8em,0em> **\dir{-}\restore\save +LD;+LU **\dir{-}\restore\save +LD;+RD-<.8em,0em> **\dir{-} \restore\save +RD+<0em,.8em>;+RU-<0em,.8em> **\dir{-} \restore \POS !UR*!UR{\cir<.9em>{r_d}};!DR*!DR{\cir<.9em>{d_l}}\restore \qw}

\newcommand{\multigate}[2]{*+<1em,.9em>{\hphantom{#2}} \qw \POS[0,0].[#1,0];p !C *{#2},p \save+LU;+RU **\dir{-}\restore\save+RU;+RD **\dir{-}\restore\save+RD;+LD **\dir{-}\restore\save+LD;+LU **\dir{-}\restore}
\newcommand{\ghost}[1]{*+<1em,.9em>{\hphantom{#1}} \qw}
\newcommand{\ustick}[1]{*!D!<0em,-.5em>=<0em>{#1}}

\newcommand{\Qcircuit}[1][0em]{\xymatrix @*=<#1>}

\newcommand{\pureghost}[1]{*+<1em,.9em>{\hphantom{#1}}}
\newcommand{\multiprepareC}[2]{*+<1em,.9em>{\hphantom{#2}}\save[0,0].[#1,0];p\save !C
  *{#2},p+RU+<0em,0em>;+LU+<+.8em,0em> **\dir{-}\restore\save +RD;+RU **\dir{-}\restore\save
  +RD;+LD+<.8em,0em> **\dir{-} \restore\save +LD+<0em,.8em>;+LU-<0em,.8em> **\dir{-} \restore \POS
  !UL*!UL{\cir<.9em>{u_r}};!DL*!DL{\cir<.9em>{l_u}}\restore}
\newcommand{\prepareC}[1]{*{\xy*+=+<.5em>{\vphantom{#1\rule{0em}{.1em}}}*\cir{l^r};p\save*!L{#1} \restore\save+UC;+UC+<.5em,0em>*!L{\hphantom{#1}}+R **\dir{-} \restore\save+DC;+DC+<.5em,0em>*!L{\hphantom{#1}}+R **\dir{-} \restore\POS+UC+<.5em,0em>*!L{\hphantom{#1}}+R;+DC+<.5em,0em>*!L{\hphantom{#1}}+R **\dir{-} \endxy}}

\newcommand{\transf}[1]{\ensuremath{\mathscr{#1}}}
\newcommand{\tA}{\transf A}
\newcommand{\tB}{\transf B}
\newcommand{\tC}{\transf C}
\newcommand{\tD}{\transf D}
\newcommand{\tE}{\transf E}
\newcommand{\tF}{\transf F}
\newcommand{\tG}{\transf G}

\newcommand{\tU}{\transf U}
\newcommand{\tV}{\transf V}

\newcommand{\tI}{\transf I}
\newcommand{\sys}[1]{\ensuremath{\mathrm{#1}}}
\newcommand{\rA}{\sys A}
\newcommand{\rB}{\sys B}
\newcommand{\rC}{\sys C}
\newcommand{\rD}{\sys D}

\newcommand{\rI}{\sys I}

\newcommand{\rX}{\sys X}
\newcommand{\rY}{\sys Y}

\newcommand{\st}{\mathsf{St}}
\newcommand{\purst}{\mathsf{PurSt}}
\newcommand{\Supp}{\mathsf{Supp}}

\newcommand{\eff}{\mathsf{Eff}}
\newcommand{\tr}{\mathsf{Tr}}

\newcommand{\R}{\mathbb{R}}
\newcommand{\N}{\mathbb{N}}

\newcommand{\bi}{\mathbf{i}}
\newcommand{\bj}{\mathbf{j}}
\newcommand{\bp}{\mathbf{p}}
\newcommand{\bPi}{\mathbf{\Pi}}
\newcommand{\be}{\mathbf{e}}
\newcommand{\Rng}{\mathsf{Rng}}
\newcommand{\Tr}{\mathrm{Tr}}

\newtheorem{theorem}{Theorem}[section]
\newtheorem{definition}{Definition}[section]
\newtheorem{proposition}{Proposition}[section]
\newtheorem{lemma}{Lemma}[section]
\newtheorem{assumption}{Assumption}

\DeclarePairedDelimiter{\norma}{\lVert}{\rVert}
\DeclarePairedDelimiter{\normop}{\lVert}{\rVert_{\mathrm {op}}}

\DeclarePairedDelimiter{\ceil}{\lceil}{\rceil}
\DeclarePairedDelimiter{\pairing}{(}{)}
\DeclarePairedDelimiter{\state}{|}{)}
\DeclarePairedDelimiter{\effect}{(}{|}
\DeclarePairedDelimiter{\abs}{|}{|}

\newcommand{\cH}{\mathcal{H}}

\newcommand{\cL}{\mathcal{L}}

\newcommand{\sX}{\mathsf{X}}
\newcommand{\sY}{\mathsf{Y}}

\def\bvec#1{\mathbf{#1}}

\usepackage{xcolor}

\begin{document}

\title{	Shannon theory beyond quantum: information content of a source}

\author{Paolo \surname{Perinotti}}

\email{paolo.perinotti@unipv.it}

\affiliation{QUIT group, Physics Dept., Pavia University, and INFN Sezione di Pavia, via Bassi 6, 27100 Pavia, Italy}

\author{Alessandro \surname{Tosini}}

\email{alessandro.tosini@unipv.it}

\affiliation{QUIT group, Physics Dept., Pavia University, and INFN Sezione di Pavia, via Bassi 6, 27100 Pavia, Italy}

\author{Leonardo \surname{Vaglini}}

\email{leonardo.vaglini01@universitadipavia.it}

\affiliation{QUIT group, Physics Dept., Pavia University, and INFN Sezione di Pavia, via Bassi 6, 27100 Pavia, Italy}

\begin{abstract}
The information content of a source is defined in terms of the minimum number of bits needed to store the output of 
the source in a perfectly recoverable way. A similar definition can be given in the case of quantum sources, with 
qubits replacing bits. In the mentioned cases the information content can be quantified through Shannon's and von 
Neumann's entropy, respectively. Here we extend the definition of information content to operational probabilistic theories, and
prove relevant properties as the subadditivity, and the relation between purity and information content of a state. We prove the consistency of the present notion of information content when applied to the classical and the quantum case. Finally, the relation with one of the notions of entropy that can be introduced in general probabilistic theories, the maximum accessible information, is given in terms of a lower bound. 
\end{abstract}
\maketitle

\section{Introduction}
The birth of information theory, marked by Shannon's pioneering work~\cite{shannon1949communication}, is represented 
by a thorough definition of information content of an information source, along with its quantification through a 
suitable quantity---the celebrated Shannon entropy. Since its early times, the information content was identified with
the minimum number of elementary information carriers---bits---needed to encode messages from the source in a 
perfectly recoverable way.

The notion of information content was then lifted to the quantum scenario, where the elementary carrier is a qubit. 
Schumacher~\cite{PhysRevA.51.2738} proved that the quantum information content of a quantum source can be 
quantified through its von Neumann's entropy. The typical setting for the proof of the above mentioned results entails
a regularisation procedure, i.e.~considering an arbitrarily large number of uses of the source, and encoding schemes 
that are not perfect, but arbitrarily accurate. This definition makes the notion of 
information content depend on the choice of the figure of merit used to evaluate accuracy ~\cite{Koashi2001,Barnum_2001,Winter_2019}. 
While in the classical 
case one uses the (complement of) error probability, Schumacher's theorem uses the entanglement fidelity, which 
evaluates how well the encoding scheme preserves entanglement of the source with an external system. As we show here, this figure of merit is equivalent to the average input-output fidelity for an arbitrary decomposition of the state 
representing the source. One big lesson of quantum information theory is thus that preserving coherence is equivalent 
to preserving entanglement, in a motto: {\em quantum information is entanglement}~\cite{schumacher_westmoreland_2010}.

In a broader sense, the above argument teaches us to judge the action of a transformation looking not only at the
way it transforms input states, but also correlations with remote systems. While in quantum theory the two 
perspectives are somehow interchangeable, there are other information theories where the second one becomes mandatory,
such as Real Quantum Theory~\cite{hardy2012limited} or Fermionic 
Theory~\cite{2014fermionic,doi:10.1142/S0217751X14300257}. The latter two theories are examples of Operational 
Probabilistic Theories~\cite{DAriano:2017aa,PhysRevA.81.062348,PhysRevA.84.012311,PhysRevA.101.042118} (OPT), or more generally of Generalized Probabilistic 
Theories~\cite{hardy1999disentangling,Spekkens:2007aa,PhysRevA.75.032304}. The definition of OPT sets a framework of 
alternative theories of elementary information carriers (or possibly elementary physical systems) and their processes, 
that can be thought of as all the conceivable information theories. In the present work we show that it is actually 
possible to extend the notion of information content discussed above to the framework of OPTs, and compare it to some 
entropic functions introduced in this context by analogy with known entropies~\cite{Barnum:2010gy,Short:2010kt,2010RpMP...66..175K}. 

The requirement that we impose on compression schemes featuring in the definition of information content is that their
effect on any preparation of ensembles that average to the considered state must be indistinguishable from leaving the
preparation untouched. Thus, besides considering any refinement of the state under discussion, we consider the action 
of the compression scheme on decompositions of its {\em dilations}, i.e.~joint 
states of the system and an arbitrary external system such that the state that one obtains after averaging and 
discarding the external system is precisely the one of interest.

The importance of considering the effect of transformations on external systems was recently discussed in other contexts, e.g.~in
assessing information vs disturbance~\cite{DAriano2020information}.

It is worth mentioning one of the main assumptions made in the present work.
 In classical and quantum Shannon theories the amount
of information is measured in bits and qubits respectively, as we have recalled just above. For a generic theory, with no further restrictions on its structure, 
in principle, one may not be able to identify an elementary information carrier. For this reason, we will consider theories that we name "digitalizable." Roughly speaking, we assume the existence of at least one elementary system, which we call obit, such that an agent can always encode an arbitrary but finite number of copies of her/his
system into an array made of an integer number of obits. This feature is evidently satisfied by classical and quantum theory, and it does not rule out scenarios that are relevant from a foundational perspective, such as non-local boxes \cite{popescu1994,d2010testing}. The latter is the prototypical 
example of a theory where the various notions of entropy exhibit odd features, such as violation of strong subadditivity, subadditivity and concavity, and where they are
also proven to be not equivalent \cite{Barnum:2010gy,Short:2010kt,2010RpMP...66..175K}. Moreover, this assumption can also be applied to theories without local
tomography, such as real quantum theory and quantum theory with superselection (e.g. fermionic theory). 

After introducing the notion of information content in a general OPT, we prove some of its main properties and
show that the optimised accessible information~\cite{Barnum:2010gy,Short:2010kt,2010RpMP...66..175K} generally 
provides a lower bound for it. As special cases, we then analyse classical and quantum theory, where our definition 
boils down to Shannon's and von Neumann's entropy, respectively. As a consequence of the present definition, finally,
Fermionic information content can be proved to coincide with von Neumann's entropy of the Fermionic 
state~\cite{perinotti2021shannon}.

The paper is organised as follows. In sec.~\ref{sec:fram}  we give an account of the OPT framework. We set up
the basic terminology and we provide the reader with the relevant definitions and assumptions that are necessary in the present work. In sec.~\ref{sec:infocont} 
we give the formal definition of information content and we show that such a definition is well-posed. We also prove properties that do not require  assumptions on the structure of the theory: subadditivity and invariance under reversible transformations acting on the state at hand. In sec.~\ref{sec:steer} we explore the consequence of the steering assumption, namely the possibility to steer any ensemble of the state by means of one of its dilations. We show that it is possible to assess the 
reliability of a compression protocol by taking into account the dilations of the state only, instead of considering also all the possible decompositions of them. We generalize the entanglement fidelity and we prove that it can be equivalently used as a figure of merit for defining the information content.  In sec. ~\ref{sec:pur} we investigate the relation between state purity and the vanishing of the information content. The main result of this section is that all the states with vanishing information content must be necessarily pure. Moreover, we show that the converse is not generally true. Indeed, we first prove
that sufficient conditions are the atomicity of parallel composition of states and essential uniqueness of purification. Then, we show that
atomicity of parallel composition of states turns out to be a necessary condition for having a null information content on pure states. Thus, any theory 
violating this property must have pure states with a strictly positive information content. Finally, we show that the accessible information is a lower bound for the information content here introduced. In sec. \ref{sec:boilsdown} we show that the information content is simply given by Shannon and von Neumann entropies in classical and quantum information theories respectively. The point here is to check the collapse of our figure of merit and the one usually used in those theories. We conclude summarizing the results of this work and we discussing some open questions in sec. \ref{sec:concl}.

\section{OPT framework}\label{sec:fram}

In this section we briefly review the framework of \emph{Operational Probabilistic Theories}.

\subsection{General description}

The primitive notions of an operational theory are those of test, event and system. A \emph{test} $\{\tA_i\}_{i\in\sX}$
is given by a collection of \emph{events}, where $i$ labels the elements
 of the outcome space $\sX$. The \emph{systems}
allow for the connection between different tests, and are denoted by capital Roman letters $\rA,\rB,\dots$. Therefore,
a test is completely determined by its input and output systems, and the events associated with the outcome space $\sX$.
In order to represent a test and its events $\{\tA_i\}_{i\in\sX}$ we use the usual diagrammatic notation
\begin{align*}
  \begin{aligned}
    \Qcircuit @C=1em @R=.7em @! R {&\ustick{\rA} \qw&\gate{\{\tA_i\}_{i\in\sX}}& \ustick{\rB}
      \qw&\qw}
  \end{aligned}
  \quad , \quad
  \begin{aligned}
    \Qcircuit @C=1em @R=.7em @! R {&\ustick{\rA}\qw&\gate{\tA_i}&\ustick{\rB} \qw &\qw}
  \end{aligned}\ ,\\
\end{align*}
and we will call $\rA$, $\rB$ the input and the output system of the test, respectively. 
If $\{\tA_i\}_{i\in\sX}$ and $\{\tB_j\}_{j\in\sY}$ are two tests,
one can define their \emph{sequential composition} as the test $\{\tC_{i,j}\}_{(i,j)\in\sX\times\sY}$, with events 
$\tC_{i,j}$ that are diagrammaticaly represented by 
\begin{align*}
  \begin{aligned}
    \Qcircuit @C=1em @R=.7em @! R {&\ustick{\rA} \qw&\gate{\tC_{i,j}}& \ustick{\rC}
      \qw&\qw}
  \end{aligned}
  \quad&\coloneqq\quad
  \begin{aligned}
    \Qcircuit @C=1em @R=.7em @! R {&\ustick{\rA}\qw&\gate{\tA_i}&\ustick{\rB} \qw &\gate{\tB_{j}}&\ustick{\rC} \qw
      &\qw}
  \end{aligned}\ .
\end{align*}
Notice that this definition requires the output system of the events on the left to be necessarily the input system of the 
events on the right. A \emph{singleton test} is a test whose outcome space set $\sX$ is a singleton, and the unique event contained
in it is called \emph{deterministic}. For any system $\rA$ there exists a unique \emph{identity} test $\{\tI_{\rA}\}$ such that 
$\tA\tI_{\rA}=\tA$ ($\tI_{\rA}\tA=\tA$) for any event $\tA$.
Another operation that can be performed on tests for defining a new test is \emph{parallel composition}. Given two systems 
$\rA$ and $\rB$ we call $\rA\rB$ the composite system of $\rA$ and $\rB$. Then, if $\{\tA_i\}_{i\in\sX}$
and $\{\tC_j\}_{j\in\sY}$ are two tests, their parallel composition is the test 
$\{\tA_i\boxtimes\tC_j\}_{(i,j)\in\sX\times\sY}$. Diagrammatically
\begin{align*}
  \begin{aligned}
    \Qcircuit @C=1em @R=1em @! R {&\ustick{\rA} \qw&\multigate{1}{\tA_i\boxtimes\tC_j}& \ustick{\rB} \qw&\qw \\
							&\ustick{\rC} \qw&\ghost{\tA_i\boxtimes\tC_j}& \ustick{\rD} \qw&\qw }
  \end{aligned}
  \quad&\coloneqq\quad
  \begin{aligned}
    \Qcircuit @C=1em @R=1em @! R {&\ustick{\rA} \qw&\gate{\tA_i}&\ustick{\rB} \qw & \qw \\
							&\ustick{\rC} \qw&\gate{\tC_{j}}&\ustick{\rD} \qw&\qw}
  \end{aligned} \, .
\end{align*}
The parallel composition operation commutes with the sequential one, namely 
$(\tE_h\boxtimes \tF_k)(\tC_i\boxtimes \tD_j)=(\tE_h \tC_i)\boxtimes (\tF_k \tD_j)$.

There is a special kind of system, the \emph{trivial system} I, satisfying $\rA\rI=\rI\rA=\rA$ for every system $\rA$. Tests with $\rI$
as input system and $\rA$ as the output one are called \emph{preparation tests} of $\rA$, while tests with input system $\rA$ and 
$\rI$ as output are named \emph{observation tests} of $\rA$. The events of a preparation test $\{\rho_i\}_{i\in\sX}$ and of an
observation test $\{a_j\}_{j\in\sY}$ are represented through the following diagrams
\begin{align*}
  \begin{aligned}
    \Qcircuit @C=1em @R=.7em @! R {\prepareC{\rho_i}& \ustick \rA
      \qw&\qw}
  \end{aligned}
  \quad&\coloneqq\quad
  \begin{aligned}
    \Qcircuit @C=1em @R=.7em @! R {&\ustick{\rI}\qw&\gate{\rho_i}&
      \ustick \rA \qw&\qw}
  \end{aligned}\quad,\\
  \\
    \begin{aligned}
    \Qcircuit @C=1em @R=.7em @! R {& \ustick \rA\qw &\measureD {a_j}}
  \end{aligned}
  \quad&\coloneqq\quad
  \begin{aligned}
    \Qcircuit @C=1em @R=.7em @! R {&\ustick{\rA}\qw&\gate{a_j}&
      \ustick \rI \qw&\qw}
  \end{aligned}\quad.
\end{align*}
In the following we will always use Greek letters to denote preparation tests and Latin letters for the observation test. Preparation and observation events
will also be denoted by using round brackets, respectively $\state{\rho}_\rA$ and $_{\rA}\effect{a}$, and we will not make explicit the system whenever it is
clear from the context.

A circuit is a diagram representing an arbitrary test that is obtained by sequential and parallel composition of other tests. We say that a circuit is \emph{closed} when the input and output systems are both the trivial one, namely, when it starts with a preparation test 
and it ends with an observation test. An \emph{operational probabilistic theory} is an operational theory where any closed
circuit (equivalently, any test from the trivial system to istelf) is given by a joint probability distribution conditioned by the tests building the circuit. Moreover, compound tests from the trivial system to istelf are independent, namely, the joint probability 
distribution is simply given by the product of the probability distributions of the composing tests. A simple example is given by a 
preparation test $\{\rho_i\}_{i\in\sX}$ sequentially followed by an observation test $\{a_j\}_{j\in\sY}$:
\begin{align*}
\{p(i,j|\{\rho_i\},\{a_j\})\}\coloneqq
\left\{
 \begin{aligned}
    \Qcircuit @C=1em @R=.7em @! R {\prepareC{\rho_i}& \ustick \rA
      \qw&\measureD{a_j}}
  \end{aligned}
\right\}  ,
\end{align*}
 with $\sum_{i,j}p(i,j)=1$. Thus, one has a joint probability distribution, which is conditioned by the chosen tests 
$\{\rho_i\}_{i\in\sX}$ and  $\{a_j\}_{j\in\sY}$. From now on we will simply omit this dependence. 
The probability associated with the closed circuit where a preparation $\rho_i$ is followed by an observation $a_j$ will also be denoted by
a pairing, $p(i,j)=(a_j|\rho_i)$.

Given any system $\rA$ of an OPT, One can define an equivalence relation on the set of preparation events by declaring that 
$\rho\sim\sigma$ iff  $\pairing{a |\rho}\neq \pairing{a |\sigma}$ for any observation event $a$. The set of equivalence classes with 
respect to this relation is called the set of \emph{states} of system $\rA$ and it is denoted by $\st(\rA)$. Similarly, one can define the 
set of \emph{effects} as the set of equivalence classes of the observation events such that $\pairing{a|\rho}\neq\pairing{b|\rho}$ for any
preparation event $\rho$, and this is denoted by $\eff(\rA)$. The sets of deterministic states and effects will be denoted
by $\st_1({\rA})$ and $\eff_1({\rA})$ respectively.

Given the probabilistic structure, states can be seen as functionals on the set of effects and viceversa, and then one can consider
linear combinations of them, thus defining two linear spaces, $\st(\rA)_{\R}$ and $\eff(\rA)_{\R}$, which are dual to each other
assuming that they are finite-dimensional. The \emph{size} $D_{\rA}$ of a given system $\rA$ is simply the dimension of the linear space $\st_{\R}(\rA)$.
A transformation event from system $\rA$ to system $\rB$ induces a linear map from $\st_{\R}(\rA\rC)$ to $\st_{\R}(\rB\rC)$ 
for any ancillary system $\rC$. 
Also the set of transformation events can be endowed with an equivalence relation. Indeed, given 
$\tA$ and  $\tB$, we say that they are \emph{operationally equivalent} ($\tA\sim\tB$) if the following identity holds
\begin{align*}
\begin{aligned}
\Qcircuit @C=0.8em @R=1.0em @! R { \multiprepareC{1}{\Psi} &\ustick \rA \qw & \gate {\tA} &\ustick  \rB \qw &\multimeasureD{1}{A}  \\
\pureghost {\Psi} &\qw & \ustick \rC  \qw   &\qw  &\ghost{A} }
\end{aligned}
\ = \
\begin{aligned}
\Qcircuit @C=0.8em @R=1.0em @! R { \multiprepareC{1}{\Psi} &\ustick \rA \qw & \gate {\tB} &\ustick  \rB \qw &\multimeasureD{1}{A}  \\
\pureghost {\Psi} &\qw & \ustick \rC  \qw   &\qw  &\ghost{A} }
\end{aligned} \ ,
\end{align*}
for any $\Psi\in\st(\rA\rC)$, $A\in\eff(\rB\rC)$ and any ancillary system $C$. In other words, two transformation events
are operationally equivalent if they induce the same linear map for any ancillary system $\rC$.
We then denote the set of all the equivalence classes with $\tr(\rA\to\rB)$, whose elements are simply called \emph{transformations}.
As for states and effects, the set of deterministic transformations will be denoted by $\tr_1(\rA\to\rB)$.
If $\tU\in\tr(\rA\to\rB)$ and there exists $\tV\in\tr(\rB\to\rA)$ such that $\tV\tU=\tI_{\rA}$ and $\tU\tV=\tI_{\rB}$, we say that $\tU$ is \emph{reversible}.
Accordingly, two systems are called {\em operationally equivalent} if there exists a reversible 
transformation $\tU\in\tr_1(\rA\to\rB)$. A notion that will be useful in the following is that of {\em asymptotically 
equivalent systems}.

\begin{definition}[Asymptotical equivalence]\label{def:asympeq}
Two systems $\rA_1$ and $\rA_2$ are asymptotically equivalent if
\begin{enumerate}
\item\label{it:kappa1}
there exists a pair of integers $k_1,k_2<\infty$, $\tE\in\tr_1(\rA_1^{\boxtimes k_1}\to\rA_2^{\boxtimes k_2})$ and $\tD\in\tr_1(\rA_2^{\boxtimes k_2}\to\rA_1^{\boxtimes k_1})$ such that $\tD\tE=\tI_{\rA_1^{\boxtimes k_1}}$;
\item\label{it:hacca2}
there exists a pair of integers $h_1,h_2<\infty$, $\tG\in\tr_1(\rA_2^{\boxtimes h_2}\to\rA_1^{\boxtimes h_1})$ and $\tF\in\tr_1(\rA_1^{\boxtimes h_1}\to\rA_2^{\boxtimes h_2})$ such that $\tF\tG=\tI_{\rA_2^{\boxtimes h_2}}$;
\item\label{it:equiciccia} let $M_{2}^{\rm min}(k_1)$ be the smallest $k_2$ such that item~\ref{it:kappa1} is satisfied for a given $k_1$, and similarly for 
$M_1^{\rm min}(h_2)$ with reference to item~\ref{it:hacca2}. The following assumption is made:
\begin{equation}
\lim_{k_1\rightarrow \infty}\frac{M_2^{\rm min}(k_1)}{k_1}=k, \qquad \lim_{h_2\rightarrow \infty}\frac{M_1^{\rm min}(h_2)}{h_2}=k^{-1}.
\label{eq:digitalizability2}
\end{equation}
\end{enumerate}
\end{definition}

Now we set up some terminology and we introduce pure and mixed states, as well as the definition of state dilation.
\begin{definition}[Refinement of an event]
Let $\tC\in\tr(\rA\to\rB)$. A refinement of $\tC$ is given by a 
collection of events $\{\tD_{i}\}_{i\in\sY}\subseteq\tr(\rA\to\rB)$
such that there exists a test $\{\tD_{i}\}_{i\in\sX}$ with $\sY\subseteq\sX$ and
$\tC=\sum_{j\in\sY}\tD_{j}$.
We say that a refinement $\{\tD_{i}\}_{i\in\sY}$ is trivial if $\tD_{i}=\lambda_{i}\tC$, $\lambda_{i}\in[0,1]$
for every $i\in\sY$. Conversely, $\tC$ is called the coarse-graining of the events $\{\tD_{i}\}_{i\in\sY}$.
\end{definition}

\begin{definition}
Given two events $\tC,\tD\in\tr(\rA\to\rB)$ we say that $\tD$ refines $\tC$, and write $\tD\prec\tC$, if there exist a refinement 
$\{\tD_i\}_{i\in\sX}$ of $\tC$ such that $\tD\in\{ \tD_i\}_{i\in\sX}$
\end{definition}

\begin{definition}[Atomic and refinable events]
An event $\tC$ is called atomic if it admits only trivial refinements. An event is refinable if it is not atomic.
\end{definition}
The notion of refinement and refinable events give rise to the definitions of \emph{pure} and \emph{mixed} states. 
\begin{definition}[Pure and mixed states]
$\rho\in\st(\rA)$ is called pure if it is atomic, or mixed otherwise. We will denote by $\purst(\rA)$ the set of all the pure 
states of system $\rA$
\end{definition}

\begin{definition}
Let $\rho\in\st(\rA)$ and $\Psi\in\st(\rA\rB)$. We say that $\Psi$ is a dilation of $\rho$ if there exists a deterministic effect $e\in\eff(\rB)$ such that
\begin{equation*}
\begin{aligned}
 \Qcircuit @C=1em @R=.7em @! R { \prepareC{\rho} & \ustick \rA \qw & \qw}
\end{aligned}
\ = \
\begin{aligned}
 \Qcircuit @C=1em @R=.7em @! R { \multiprepareC{1}{\Psi} & \ustick \rA \qw& \qw \\
 \pureghost {\Psi} & \ustick \rB \qw&\measureD{e}}
\end{aligned}\ .
\end{equation*}
We denote by $D_{\rho}$ the set of all dilations of the state $\rho$. If $\Psi$ is also pure, then we say
that it is a purification of $\rho$ and $\rB$ is called the purifying system. Finally, we denote by $P_{\rho}$ the set of all
the purifications of $\rho$.
\end{definition}
Trivially one has that $P_{\rho}\subseteq D_{\rho}$. Moreover, if $\Omega\in D_{\rho}$, then one has $D_{\Omega}\subseteq D_{\rho}$

Two special instances of this framework are classical and quantum theory.
In classical theory, the systems are associated with real vector spaces $\R^{d_A}$, and different systems are associated with different values of $d_A$. The set of states is made of sub-stochastic vectors
in these spaces, namely by vectors ${\bvec x}$ satisfying $\norma{\bvec x}_1 \coloneqq\sum_{i=1}^{d_A}\abs{x_i}\leq 1$,
and therefore it is a simplex. The pure states are then represented by the canonical basis vectors ${\bvec e}_i$ of 
$\R^{d_A}$. The convex set of effects is given by unit-dominated positive vectors, i.e. those $\bvec x \in \R^{d_{\rA}}$ such that $0\leq x_i\leq 1$
for any $i=0,\dots,d_A$. Transformations from system $\rA$ to system $\rB$ are represented by $d^{\rB}\times d^{\rA}$ 
sub-stochastic matrices $\bvec M$ acting on the probability vectors by multiplication ${\bvec x}\rightarrow{\bvec M \bvec x}$. 
Recall that a matrix is sub-stochastic when each column is a sub-stochastic vector. Sequential and parallel composition are trivially
given by matrix multiplication and tensor product respectively.

In quantum theory systems are associated with Hilbert Spaces, and different systems are represented by spaces of different
dimension $d$ (we will assume $d<\infty$ whenever we will refer to results relative to quantum theory). The convex set of states of
a system $\rA$ is given by sub-normalized density matrices $\rho$ on the associated space $\cH_{\rA}$, i.e. matrices such that 
$\rho>0$ and $\Tr(\rho)\leq 1$. The convex set of effects is given by functionals (acting on the set of states)
of the form $\Tr(\cdot E)$, where $E$ is a positive matrix dominated by the identity, $0\leq E \leq I$. Transformations in
$\tr(\rA\to\rB)$ are mathematically represented by Completely Positive Trace non-Icreasing maps. These have a Kraus decomposition,
namely, if $\tE\in\tr(\rA\to\rB)$ there exists a set of operators $\{E_i\}\subseteq\cL(\cH_{\rB},\cH_{\rA})$ such that
$\tE(\cdot)=\sum_i E_i^\dagger \cdot E_i$. Sequential composition is simply given by composition of maps, and the parallel one
is represented by the tensor product, as in the classical case.

The linear space $\st_{\R}(\rA)$ can be endowed with a metric structure by means of the following norm,
 which has an operational
meaning related to optimal discrimination schemes~\cite{PhysRevA.81.062348}.
\begin{definition}[Operational norm]
The norm of an element $\rho\in\st(\rA)_{\R}$ is defined as
\[
\norma{\rho}_{\rm op}\coloneqq\sup_{a\in\eff(\rA)}(2a-e_{\rA}|\rho),
\]
where $e_{\rA}$ is the deterministic effect obtained by the coarse-graining of the observation test containing $a$.
\end{definition}
This norm satisfies a monotonicity property, as stated in the following lemma
\begin{lemma}[Monotonicity of the operational norm]
For any $\delta\in\st_{\R}(\rA)$ and $\tC\in\tr(\rA,\rB)$ the following inequality holds
\begin{equation}
\normop{\tC\delta}\leq\normop{\delta},
\end{equation}
with the equality holding iff $\tC$ is reversible.
\end{lemma}
This notion of norm, which is valid for any OPT, reduces to the trace-norm in quantum theory.

\subsection{Restricting the class of theories}

Upon marginalization over the observation test, one can define the preparation probability conditioned by the test $\{a_j\}_{j\in\sY}$
as $p(\rho_i|\{\ a_j\})\coloneqq\sum_j p(i,j)$. Generally, the preparation probability is not one, unless the preparation test 
$\{\rho_i\}_{i\in\sX}$ is the singleton, i.e. the state is deterministic. Moreover, as it is clear by its definition, it can also depend on the observation test we are marginalizing over. 
Usually, the causality condition is expressed as a \emph{no-signalling from the future} principle, namely by saying that preparation
probabilities are actually independent of the chosen observation test, which is equivalent to state the uniqueness of the determinstic
effect. In this paper we will adopt a stronger form of causality, that is the following one
\begin{assumption}[Causal theories]
An OPT satisfies strong causality if for every test $\{\tA_{i}\}_{i\in \mathsf{X}}$ and every collection of tests $\{\tB_{j}^{i}\}_{j\in\sY}$
labelled by $j\in\sY$, the collection of events $\{\tC_{i,j}\}_{(i,j)\in\sX\times\sY}$ with
\begin{align*}
  \begin{aligned}
    \Qcircuit @C=1em @R=.7em @! R {&\ustick{\rA} \qw&\gate{\tC_{i,j}}& \ustick{\rC}
      \qw&\qw}
  \end{aligned}
  \quad&\coloneqq\quad
  \begin{aligned}
    \Qcircuit @C=1em @R=.7em @! R {&\ustick{\rA}\qw&\gate{\tA_i}&\ustick{\rB} \qw &\gate{\tB_{j}^{i}}&\ustick{\rC} \qw
      &\qw}
  \end{aligned}\ ,\\
\end{align*}
is a test of the theory.
\end{assumption}
One can show that the above statement implies uniqueness of the deterministic effect~\cite{DAriano:2017aa,PhysRevA.81.062348}. 

Another assumption that we will use in some sections of this work stems from the steering property of quantum theory.
This asserts that, given a state $\rho\in\st(\rA)$ and a purification $\Phi\in\purst(\rA\rB)$ of $\rho$, for any decomposition 
$\sum_{i\in\sX} p_i\sigma_i$ of $\rho$ there exists an observation test $\{ b_i\}_{i\in\sX}\subseteq\eff(\rB)$ such that
\begin{equation*}
\  \
\begin{aligned}
 \Qcircuit @C=1em @R=.7em @! R { \prepareC{\sigma_{i}} & \ustick \rA \qw & \qw}
\end{aligned}
\ =\
\begin{aligned}
 \Qcircuit @C=1em @R=.7em @! R { \multiprepareC{1}{\Phi} & \ustick \rA \qw& \qw \\
 \pureghost {\Phi} & \ustick \rB \qw&\measureD{b_{i}}}
\end{aligned}\ .
\qquad \forall i\in \sX.
\end{equation*}
This feature cannot be assumed as it stands: the first reason is that a generic OPT may not encompass the existence of
a purification for any state of the theory (and classical theory is a trivial example), therefore we are led to consider
dilations instead of purifications. Secondly, there is no reason why one should be able to steer any decomposition by means of
the same dilation. Thus, for a generic theory, one can state the steering as follows.
\begin{assumption}[Steering]\label{ass:steering}
Let $\rho\in\st(\rA)$ and $\{\sigma_{i}\}_{i\in \sX}\subseteq\st(\rA)$ be a refinement of $\rho$. Then there exist a system $\rB$, a state $\Psi\in\st(\rA\rB)$, and an
observation test $\{b_{i}\}_{i\in \sX}$ such that
\begin{equation*}
\  \
\begin{aligned}
 \Qcircuit @C=1em @R=.7em @! R { \prepareC{\sigma_{i}} & \ustick \rA \qw & \qw}
\end{aligned}
\ =\
\begin{aligned}
 \Qcircuit @C=1em @R=.7em @! R { \multiprepareC{1}{\Psi} & \ustick \rA \qw& \qw \\
 \pureghost {\Psi} & \ustick \rB \qw&\measureD{b_{i}}}
\end{aligned}\ ,
\qquad \forall i\in \sX.
\end{equation*}
\end{assumption}
Notice that the state $\Psi$ in the steering assumption must be a dilation of $\rho$, as one can easily verify upon
summing over $i\in\sX$. The stronger steering feature satisfied by quantum theory can actually be proven to hold
in any OPT satisfying atomicity of parallel composition of states,
existence and uniqueness (up to reversible channels) of purification and perfect discriminability as axioms. For the present purposes we choose to state
steering as a property that an OPT may or may not satisfy rather than discussing the conditions under which
it holds.

Another property that we will assume throughout the paper is digitalizability.

\begin{assumption}[Digitalizability]\label{ass:digitalizability}
We say that an OPT is digitalizable if there exists a system $\rB$ (called \emph{obit}) such that for any system $\rX$ there
exists $k<\infty$ and a pair of maps $\tC\in\tr_1(\rX\to\rB^{\boxtimes k})$ and 
$\tF\in\tr_1(\rB^{\boxtimes k}\to\rX)$ such that $\tF\circ\tC=\tI_{\rX}$. Moreover, if $\rB_1$ and $\rB_2$ are two such systems, then they are asymptotically equivalent.
\end{assumption}

The above assumption holds in quantum theory, since any qudit system can be taken as elementary. Moreover, let us consider two different qudits with
dimension $d_1$ and $d_2$ respectively. Generally, the equation $d_1^N=d_2^M$ may have no integer solutions (for instance, when both $d_1$ and $d_2$ are
prime). However, the smallest integer $M$ such that we can isometrically embed $\cH_1^{\otimes N}$ into $\cH_2^{\otimes M}$ is given by 
$\ceil{N\log_{d_2} d_1}$ which is such that $\lim_{N \rightarrow \infty}\frac{\ceil{N\log_{d_2} d_1}}{N}=\log_{d_2}d_1$. Similarly, $\ceil{M\log_{d_1} d_2}$
is needed for an isometric embedding of $\cH_2^{\otimes M}$ into $\cH_1^{\otimes N}$, and $\lim_{N \rightarrow \infty}\frac{\ceil{M\log_{d_1} d_2}}{M}=\log_{d_1}d_2(=1/\log_{d_2}d_1)$.
In any theory satisfying this assumtpion we can always encode the state of our system on the parallel composition of a sufficiently large
number of elementary systems, which we can think of as a generalization of the qubit system for quantum theory. The request of digitalizability
comes from the need of a unit for the amount of information required for storing a given source. In
classical information theory we use bits, in the quantum counterpart the qubits, and for a generic OPT satisfying digitalizability we
use obits, whose existence must then be postulated. 
We want to stress the fact that this hypothesis is not so restrictive on the class of possible theories.
The non-local boxes \cite{popescu1994} provides us with a non trivial example of a theory
satisfying the digitalizability assumption, and with a strong departure from the
quantum one. There exists a unique single system, whose state space is described by a square, and multipartite systems are obtained by using only this one, therefore it is trivially digitalizable.  It is in this scenario that the diversity of the known entropy notions is manifest. Moreover, the fact that we are only referring to the conversion of finitely many copies of the system at hand is not
so constraining from a conceptual point of view. For the present purposes, namely taking a first step
towards a Shannon theory for generic physical systems, this level of analysis is sufficient. However, the composition of a countable number of systems can be suitably defined (see \cite{Perinotti2020cellularautomatain}) opening the route to a generalization of this property in the infinite case.

\section{Information Content in OPT}\label{sec:infocont}

Let $\rho\in\st_{1}(\rA)$ and consider $N$ copies of the system on which we have prepared the same state $\rho$
and let $M$ be a positive integer. A compression scheme is then a pair of maps 
$\tE\in\tr_{1}(\rA^{\boxtimes N}\to\rB^{\boxtimes M})$, $\tD\in\tr_{1}(\rB^{\boxtimes M}\to\rA^{\boxtimes N})$.

\begin{definition}An $(\varepsilon,N)$-reliable compression scheme $(\tE,\tD)$ is such that
\[
\sup_{\rC,\{\Psi_i\}}\sum_{i\in\sX}\norma{[(\tD\circ\tC)\boxtimes\tI_\rC]\Psi_{i}-\Psi_{i}}<\varepsilon, 
\] 
where $\{\Psi_{i}\}_{i\in\sX}\subseteq\st_1(\rA^{\boxtimes N}\rC)$ denotes a refinement of any dilation of $\rho^{\boxtimes N}$.
For fixed $N$, $M$ we denote with $E_{N,M,\varepsilon}(\rho)$ the set of $\varepsilon$-reliable
compression schemes. 
\end{definition}

\begin{definition}[Information Content]
Let $\rho\in\st_1(\rA)$. We define the {\em smallest achievable compression ratio} for length $N$ to tolerance $\varepsilon$ as follows
\begin{equation}\label{eq:inforate}
R_{N,\varepsilon}(\rho)\coloneqq\frac{\min\{ M:E_{N,M,\varepsilon}(\rho)\neq\emptyset\}}{N}.
\end{equation}
The {\em information content} of the state $\rho$ is defined as
\begin{equation}\label{eq:infocontent}
I(\rho)\coloneqq\lim_{\varepsilon\rightarrow 0}\limsup_{N\rightarrow \infty}R_{N,\varepsilon}(\rho).
\end{equation}
\end{definition}
\begin{proposition}
$I(\rho)$ is well defined for every $\rho\in\st(\rA)$ and every system $\rA$.
\end{proposition}

\begin{proof}
Firstly, we show that for any choice of the elementary system, $I(\rho)$ is a finite number for any state $\rho$.
By the digitalizability assumption we know that for any $N$ there exists a positive integer $K<\infty$ and a pair of maps
$\tE\in\tr_1(\rA^{\boxtimes N}\to\rB^{\boxtimes K})$, $\tD\in\tr_1(\rB^{\boxtimes K}\to\rA^{\boxtimes N})$  such that 
$\tD\tE=\tI_{\rA^{\boxtimes N}}$.
Therefore, for any $N$, $\varepsilon$ the set $E_{N,K,\varepsilon}(\rho)$ is not
empty, and the minimum in eq.~\eqref{eq:inforate} is always finite. Moreover, it is immediate to realise that $K$ does not need to grow more than linearly versus $N$, just considering $N$ repetitions of the encoding for one copy $\rA$.
Thus, the ratio in eq.~\eqref{eq:inforate} is bounded and one can take the $\limsup_{N\rightarrow\infty}$ safely. The existence of $\lim_{\varepsilon\rightarrow 0}$ follows 
by the fact that $E_{N,M,\varepsilon}(\rho)\subseteq E_{N,M,\varepsilon '}(\rho)\neq\emptyset$
whenever $\varepsilon\leq\varepsilon '$, which, in turn, implies monotonicity of the function 
$\limsup_{N\rightarrow\infty}R_{N,\varepsilon}(\rho)$ versus $\varepsilon$.

What is left to prove is that using two different obits we are not led to two incomparable notions of information content. 
First of all, fix $N,\varepsilon$, let $\rho\in\st_1(\rA)$, and let
\begin{align*}
&M_{1,N}\coloneqq\min\{M:E^1_{N,M,\varepsilon}(\rho)\neq\emptyset\}, \\
&M_{2,N}\coloneqq\min\{M:E^2_{N,M,\varepsilon}(\rho)\neq\emptyset\},
\end{align*}
be the minimum number of obits $\rB_1$ and $\rB_2$ needed for an $\varepsilon$-optimal encoding, respectively.  Rephrasing the first equation in \eqref{eq:digitalizability2}, there exists a sequence $\delta_1(M_1)$ 
such that $M_2^{\rm min}(M_1)=kM_1+\delta_1(M_1)$ with $\lim_{M_1\rightarrow \infty}\frac{\delta_1(M_1)}{M_1}=0$. Given the  
encoding $\tE$ of item~\ref{it:kappa1} in definition~\ref{def:asympeq} from $M_{1,N}$ to $M_2^{\rm min}(M_{1,N})$ 
(see item~\ref{it:equiciccia} in definition~\ref{def:asympeq}), we have an $\varepsilon$-optimal encoding of $\rho^{\boxtimes N}$ onto $M_2^{\rm min}(M_{1,N})$
obits $\rB_2$, therefore $M_{2,N}\leq M_2^{\rm min}(M_{1,N})$ and this implies
\begin{align*}
\limsup_{N\rightarrow \infty}&\frac{M_{2,N}}{N}\leq \limsup_{N\rightarrow \infty}\frac{M_2^{\rm min}(M_{1,N})}{N}\\
&=\limsup_{N\rightarrow \infty}\left[\frac{kM_{1,N}}{N}+\frac{\delta_1(M_{1,N})}{N}\right] \\
&\leq\limsup_{N\rightarrow \infty}\frac{kM_{1,N}}{N}+\limsup_{N\rightarrow \infty}\left|\frac{\delta_1(M_{1,N})}{M_{1,N}}\frac{M_{1,N}}{N}\right| \\
&=k\limsup_{N\rightarrow \infty}\frac{M_{1,N}}{N}.
\end{align*}
The last line follows by $\lim_{M_1\rightarrow \infty}\frac{\delta_1(M_1)}{M_1}=0$ along with the fact that $M_{1,N+1}\geq M_{1,N}$. Taking 
$\lim_{\varepsilon \rightarrow 0}$ we end up with $I_2(\rho)\leq kI_1(\rho)$. A similar argument can be used to show the reverse inequality, and we have
that, for any $\rho\in\st(\rA)$, $I_2(\rho)=kI_1(\rho)$.
\end{proof}

\begin{proposition}[Subadditivity]\label{prop:subadditivity}
Let $\Psi\in\st_1(\rA\rB)$ and let $\rho\in\st_1(\rA)$, $\sigma\in\st_1(\rB)$ be its marginals. Then the following property holds
\[
I(\Psi)\leq I(\rho)+I(\sigma).
\]
\end{proposition}
\begin{proof}
Let  $(\tE^{\rho},\tD^{\rho})\in E_{N,\overline{K},\varepsilon}(\rho)$,
$(\tE^{\sigma},\tD^{\sigma})\in E_{N,\overline{L},\varepsilon}(\sigma)$ with 
$\overline{K}\coloneqq\min \{M: E_{N,M,\varepsilon}(\rho)\neq\emptyset\}$ and similarly for $\overline{L}$. Now let 
$\{ \Gamma_i \}_{i\in\sX}$ be such that $\sum_{i\in\sX}\Gamma_i \in D_{\Psi^{\boxtimes N}}$ and consider 
$\{ (\tC^{\rho}\boxtimes \tI) (\Gamma_{i})\}_{i\in\sX}$, where $\tC^{\rho}\coloneqq\tD^{\rho}\tE^{\rho}$. Since
$\tC^{\rho}$ is a channel and $D_{\Psi^{\boxtimes N}}\subseteq D_{\rho^{\boxtimes N}}$,
$D_{\Psi^{\boxtimes N}}\subseteq D_{\sigma^{\boxtimes N}}$ we have that 
$\sum_{i\in\sX}(\tC^{\rho}\boxtimes \tI) (\Gamma_{i})$ is a dilation of both  $\rho^{\boxtimes N}$ and
$\sigma^{\boxtimes N}$. Similarly for $\tC^\sigma\coloneqq\tD^\sigma\tE^\sigma$. This implies the following bound
\begin{align*}
& \sum_{i\in\sX}\norma{(\tC^{\rho}\boxtimes \tC^{\sigma}\boxtimes \tI)(\Gamma_i)-\Gamma_i}_{op}\\
=&  \sum_{i\in\sX}\norma{(\tC^{\rho}\boxtimes \tC^{\sigma}\boxtimes \tI)(\Gamma_i)-(\tC^{\rho}\boxtimes \tI)(\Gamma_i) \\
					&+(\tC^{\rho}\boxtimes \tI)(\Gamma_i)-\Gamma_i}_{op} \\
< \ & 2\varepsilon  ,
\end{align*}
where we used the triangle inequality for the operational norm. Thus $E_{N,\overline{K}+\overline{L},2\varepsilon}(\Psi)\neq\emptyset$ and this implies that
\[
\frac{\min\{ M:E_{N,M,2\varepsilon}(\Psi)\neq\emptyset \}}{N} \leq\ \frac{\overline{K}}{N}+\frac{\overline{L}}{N}.
\]
And finally, by taking the $\limsup_{N\rightarrow\infty}$ and then $\lim_{\varepsilon\rightarrow 0}$ on both sides
we get the thesis.
\end{proof}

We notice that, in order to compute the information content, one can test the compression
schemes on pure decompositions $\{p_i,\Phi_i\}$ only. More precisely, let $E_{N,M,\varepsilon}^\mathrm{pur}(\rho)$ be the set of schemes which are $(\varepsilon,N)$-reliable according to the following criterion
\begin{equation}\label{eq:pure_figure_merit}
\sup\left(
\sum_{i}p_i\normop{(\tD\tE\boxtimes\tI)(\Phi_i)-\Phi_i}\right)<\varepsilon,
\end{equation}
where the supremum is taken on all the pure decompositions $\{p_i,\Phi_i\}$ of any $\Omega\in D_{\rho^{\boxtimes N}}$. Let $I^\mathrm{pur}(\rho)$ be the information content computed restricting to such maps, i.e.
\begin{align*}
&R^\mathrm{pur}_{N,\varepsilon}(\rho)\coloneqq\frac{\min\{ M:E^\mathrm{pur}_{N,M,\varepsilon}(\rho)\neq\emptyset\}}{N}\\
&I^\mathrm{pur}(\rho)\coloneqq\lim_{\varepsilon\rightarrow 0}\limsup_{N\rightarrow \infty}R^\mathrm{pur}_{N,\varepsilon}(\rho).
\end{align*}
Then one has $I(\rho)=I^\mathrm{pur}(\rho)$.

\begin{lemma}\label{lemma:pur_infocontent}
Let $\rho \in \st_1(\rA)$, then $I(\rho)=I^\mathrm{pur}(\rho)$
\end{lemma}
\begin{proof}
On the one hand, we trivially have $I^\mathrm{pur}(\rho)\leq I(\rho)$. On the other hand, let $\{\Psi_i\}_{i\in \sX}$ be a refinement of $\Omega\in D_{\rho^{\boxtimes N}}$. For any $i$ we can further decompose $\Psi_i$ in
terms of pure states $\{q_{j}^{i},\Phi_{i,j}\}_{j\in\sY}$, with $\sum_{i,j}q_{j}^{i}=1$. Therefore $\{q_{j}^{i},\Phi_{i,j}\}_{(i,j)\in\sX\times\sY}$ is
a pure state decomposition of $\Omega$, and by the triangle inequality one has
\begin{align*}
&\sum_{i\in\sX}\normop{[(\tD\circ\tC)\boxtimes\tI]\Psi_{i}-\Psi_{i}}\\ 
&\leq\sum_{(i,j)\in\sX\times\sY}q_{j}^{i}\normop{[(\tD\circ\tC)\boxtimes\tI]\Phi_{i,j}-\Phi_{i,j}}.
\end{align*}
This implies that $E_{N,M,\varepsilon}^\mathrm{pur}(\rho)\subseteq E_{N,M,\varepsilon}(\rho)$, and in turns that $I(\rho)\leq I^\mathrm{pur}(\rho)$ Therefore
$I(\rho)=I^\mathrm{pur}(\rho)$.
\end{proof}

\begin{proposition}Let $\rho\in\st_1(\rA)$ and $\tU \in \tr_1(\rA)$ be a reversible channel, then $I(\rho)=I(\tU(\rho))$.
\end{proposition}
\begin{proof}
We show that  $E_{N,M,\varepsilon}(\tU(\rho))\neq\emptyset\implies E_{N,M,\varepsilon}(\rho)\neq\emptyset$.
Let $(\tE,\tD)\in  E_{N,M,\varepsilon}(\tU(\rho))$ and let $\{\Psi_{i}\}_{i\in\sX}$ be such that 
$\sum_{i\in\sX}\Psi_{i}\in D_{\rho^{\boxtimes N}}$. It is clear that 
$\sum_{i\in\sX}( \tU^{\boxtimes N}\boxtimes\tI)(\Psi_{i})\in D_{\tU(\rho)^{\boxtimes N}}$ and therefore
\[
\sum_{i\in\sX}\normop{[( \tD\tE-\tI)\boxtimes\tI]( \tU^{\boxtimes N}\boxtimes\tI)(\Psi_{i})}<\varepsilon.
\]
Upon defining $\tilde{\tE}\coloneqq\tE\tU^{\boxtimes N}$ and $\tilde{\tD}\coloneqq(\tU^{-1})^{\boxtimes N}\tD$, recalling that $\tU$ is reversible and that the
operational norm is invariant under reversible transformations, the above inequality can be rewritten as follows
\begin{align*}
\varepsilon>&\sum_{i\in\sX}\normop{( \tU^{\boxtimes N}\boxtimes\tI)[(\tilde{\tD}\tilde{\tE}\boxtimes\tI)(\Psi_{i})-(\Psi_{i})]}\\
		=&\sum_{i\in\sX}\normop{[(\tilde{\tD}\tilde{\tE}\boxtimes\tI)(\Psi_{i})-(\Psi_{i})]},
\end{align*}
namely, since $\{\Psi_{i}\}_{i\in\sX}$ is arbitrary, $(\tilde{\tE},\tilde{\tD})\in E_{N,M,\varepsilon}(\rho)\neq\emptyset$.
This implies that $R_{N,\varepsilon}(\rho)\leq R_{N,\varepsilon}(\tU(\rho))$, and then $I(\rho)\leq I[\tU(\rho)]$.
The reverse inequality is now trivial
\[I(\rho)=I[\tU^{-1}\tU(\rho)]\geq I[\tU(\rho)],\]
where we have used the previous result along with the fact that $\tU^{-1}$ is also reversible.
\end{proof}
\section{Steering: Information Content from Dilations}\label{sec:steer}
In this section we show that in an OPT satisfying the steering property (assumption~\ref{ass:steering}) the information content of a state can be computed by considering the action of the compression schemes on the set $D_{\rho^{\boxtimes N}}$ only.
\begin{lemma}\label{lem:dillemma}
Let $\rho\in\st_1(\rA)$, $\tC\in\tr_{1}(\rA)$ and $\varepsilon>0$. If $\norma{\tC\boxtimes\tI(\Psi)-\Psi}<\varepsilon$
for any $\Psi\in D_{\rho}$ and assumption~\ref{ass:steering} holds, then one has
\[
\sum_{i\in\sX}\normop{\tC\boxtimes\tI(\Sigma_{i})-\Sigma_{i}}<\varepsilon,
\]
for any refinement $\{\Sigma_{i}\}_{i\in\sX}$ of an element of $D_{\rho}$.
\end{lemma}
\begin{proof}
Let $\{\Sigma_{i}\}_{i\in\sX}$ be the refinement of an element $\Omega$ of $D_{\rho}$. By assumption~\ref{ass:steering} there
exists $\Gamma\in D_{\Omega}\subseteq D_{\rho}$ and an observation test $\{c_i\}_{i\in\sX}$ such that 
\begin{equation*}
\  \
\begin{aligned}
 \Qcircuit @C=1em @R=.7em @! R { \multiprepareC{1}{\Sigma_i} & \ustick \rA \qw& \qw \\
 \pureghost {\Sigma_i} & \ustick \rB \qw& \qw}
\end{aligned}
\ =\
\begin{aligned}
 \Qcircuit @C=1em @R=.7em @! R { \multiprepareC{2}{\Gamma} & \ustick \rA \qw& \qw \\
 		\pureghost {\Gamma} & \ustick \rB \qw& \qw \\ 
		\pureghost {\Gamma} & \ustick \rC \qw&\measureD{c_{i}}}
\end{aligned}\ ,
\qquad \forall i\in \sX.
\end{equation*}
For any $i\in\sX$, let $A_{i}\in\eff(\rA\rB)$ be the effect achieving the norm, namely such that
\[
\normop{\tC\boxtimes\tI(\Sigma_{i})-\Sigma_{i}}=(A_{i}|[(\tC-\tI)\boxtimes\tI]|\Sigma_{i}).
\]
Since $\sum_{i\in\sX}A_i\boxtimes c_i$ is an effect (see appendix~\ref{app:trivial_lemma}), we have that 
\begin{align*}
&  \sum_{i\in\sX} \normop{\tC\boxtimes\tI(\Sigma_{i})-\Sigma_{i}} =\\
& \sum_{i\in\sX} \
\begin{aligned}
 \Qcircuit @C=0.7em @R=.4em @! R { \multiprepareC{2}{\Gamma} & \ustick \rA \qw&\gate{\tC}& \ustick \rA \qw&\multimeasureD{1}{A_i} \\
 		\pureghost {\Gamma} & \ustick \rB \qw& \qw  &\qw&\ghost{A_i}\\ 
		\pureghost {\Gamma} & \qw& \ustick \rC \qw&\qw &\measureD{c_{i}}}
\end{aligned}
\ - \
\begin{aligned}
 \Qcircuit @C=1em @R=.7em @! R { \multiprepareC{2}{\Gamma} & \ustick \rA \qw&\multimeasureD{1}{A_i} \\
 		\pureghost {\Gamma} & \ustick \rB \qw&\ghost{A_i} \\ 
		\pureghost {\Gamma} & \ustick \rC \qw&\measureD{c_{i}}}
\end{aligned} \\
& \leq  \normop{\tC\boxtimes\tI(\Gamma)-\Gamma} < \varepsilon,
\end{align*}
which straightforwardly leads to the thesis.
\end{proof}

\begin{proposition}\label{prop:infocontent_dil}
Let $\rho\in\st_1(\rA)$ and consider
\begin{equation}
I^\mathrm{dil}(\rho)\coloneqq\lim_{\varepsilon\rightarrow 0}\limsup_{N\rightarrow \infty} \frac{\min\{M:E^\mathrm{dil}_{N,M,\varepsilon}(\rho)\neq\emptyset\}}{N},
\end{equation}
where $E^\mathrm{dil}_{N,M,\varepsilon}(\rho)$
is the set of the compression schemes $(\tE,\tD)$ such that 
$\sup_{\Psi\in D_{\rho^{\boxtimes N}}}\norma{\tC\boxtimes\tI(\Psi)-\Psi}<\varepsilon$.
Then $I(\rho)=I^\mathrm{dil}(\rho)$.
\end{proposition}
\begin{proof}
This is a straightforward consequence of the lemma. Indeed, all the dilations are a refinement of themselves, so that
$E_{N,M,\varepsilon}(\rho)\subseteq E^\mathrm{dil}_{N,M,\varepsilon}(\rho)$ and therefore $I^\mathrm{dil}(\rho)\leq I(\rho)$. The inclusion
$E_{N,M,\varepsilon}(\rho)\supseteq E^\mathrm{dil}_{N,M,\varepsilon}(\rho)$ follows by lemma~\ref{lem:dillemma}, and this implies the thesis.
\end{proof}

We now prove some bounds that involve a quantity generalising the classical and quantum fidelity. For this purpose, we consider a definition of fidelity \cite{2010RpMP...66..175K} that can be adopted in the OPT framework, which reduces to the classical or quantum one in the in the respective theories.

\begin{definition} \label{def:fidelity_OPT}
Let $\rho,\sigma\in\st_1(\rA)$. For any observation test $\{a_{i}\}_{i\in \sX}\subseteq \eff{(\rA)}$ denote by ${\bf p}\coloneqq p_{i}$
and ${\bf q}\coloneqq q_{i}$ the probability distributions defined by
\begin{align*}
p_{i}\coloneqq& \
  \begin{aligned}
    \Qcircuit @C=1em @R=.7em @! R {\prepareC{\rho}& \ustick \rA
      \qw& \measureD{a_{i}}}
  \end{aligned}\ ,
\\
q_{i}\coloneqq& \
  \begin{aligned}
    \Qcircuit @C=1em @R=.7em @! R {\prepareC{\rho}& \ustick \rA
      \qw&\measureD{a_{i}}}
  \end{aligned}\ .
\end{align*}
Then one can define the Fidelity between $\rho$ and $\sigma$ as
\begin{equation}
F(\rho,\sigma)\coloneqq\inf_{\{a_{i}\}_{i\in\sX}}F_{c}({\bf p},{\bf q}), \label{eq:fidelity_OPT}
\end{equation}
where $F_{c}({\bf p},{\bf q})=\sum_{i}\sqrt{p_{i}q_{i}}$.
\end{definition}

Since the classical fidelity is bounded by 1, and is equal to 1 only for $\bf p=\bf q$, one clearly has $F(\rho,\sigma)\leq1$, with equality if and only if $\rho=\sigma$. Fidelity satisfies the following property, that generalises the Fuchs-van de Graaf inequality~\cite{761271} and is relevant for the present work.

\begin{proposition}\label{prop:gen_fuchs}
Let $\rho,\sigma\in\st_1(\rA)$. The following inequalities hold
\begin{equation}\label{eq:gen_fuchs}
1-F(\rho,\sigma)\leq\frac{1}{2}\normop{\rho-\sigma}\leq\sqrt{1-F(\rho,\sigma)^2}.
\end{equation}
\end{proposition}

\begin{proposition}[Monotonicity]\label{prop:fid_monot}
Let $\rho,\sigma\in\st_1(\rA)$ and $\tC\in\tr_{1}(\rA,\rB)$. The following inequality holds
\begin{equation}
F(\tC(\rho),\tC(\sigma))\geq F(\rho,\sigma).
\end{equation}
\end{proposition}

In QT we have a notion, the so called \emph{entanglement fidelity}, which measures how well correlations 
with an environment are preserved by a given channel acting on our local system. If $\rho\in\st_1(\rA)$ is the state of our local system,
 $\Phi\in\purst(\rA\rB)$ is a purification of $\rho$ and $\tC\in\tr_1(\rA\to\rC)$ the channel locally applied to $\rA$,
then the entanglement fidelity is defined as the square of the Uhlmann one between input and output, i.e.~$F[\Phi,\tC\otimes\tI(\Phi)]^2$. This is a well defined quantity since it is independent of the chosen purification. By means of 
the generalized notion of fidelity in~\eqref{eq:fidelity_OPT} we can define an analogous of the entanglement fidelity in the OPT framework.
Again, we must be aware of the fact that in a generic OPT, there may be states that cannot be purified (mixed states in classical theories are a rather trivial example). In order to encompass the most general situation, we refer to dilations rather than focusing on purifications. Moreover, we want to define a quantity that is independent of the particular dilation, and we are thus led to the following definition.
\begin{definition}\label{def:ent_fid}
Let $\rho\in\st(\rA)$ and $\tC\in\tr_{1}(\rA\to\rC)$. We define the correlation fidelity as follows
\begin{equation}
F(\rho,\tC)=\inf_{\Psi\in D_{\rho}}F[\Psi,\tC\boxtimes\tI(\Psi)]^{2}.
\end{equation}
\end{definition}
By means of the generalized Fuchs-van de Graaf inequality~\eqref{eq:gen_fuchs} we can see that the correlation fidelity can be used as
an equivalent figure of merit on $D_{\rho^{\boxtimes N}}$. More precisely, the following proposition holds.
\begin{proposition}\label{prop:infocontent_fidelity}
Let $\rho\in\st_1(\rA)$ and define
\[
I^{F}(\rho)\coloneqq\lim_{\varepsilon\rightarrow 0}\limsup_{N\rightarrow \infty} \frac{\min\{M:E^{F}_{N,M,\varepsilon}(\rho)\neq\emptyset\}}{N},
\]
where
\[
E_{N,M,\varepsilon}^{F}(\rho)\coloneqq\{(\tE,\tD)|F(\rho^{\boxtimes N},\tD\tE)>1-\varepsilon \},
\]
then $I^\mathrm{dil}(\rho)=I^{F}(\rho)$.
\end{proposition}

\begin{proof}
This is a straightforward consequence of proposition~\ref{prop:gen_fuchs}. Let $(\tE,\tD)\in E^\mathrm{dil}_{N,M,\varepsilon}(\rho)$.
By the first inequality in~\eqref{eq:gen_fuchs} we have that 
\[
F[\Psi,(\tC\boxtimes\tI) (\Psi)]^{2}\geq 1 -\varepsilon + \frac{\varepsilon^2}{4}\geq 1 - \varepsilon,
\]
for any $\Psi\in D_{\rho^{\boxtimes N}}$, and this implies
\[
F(\rho^{\boxtimes N},\tC)\geq 1 - \varepsilon.
\]
Therefore $(\tE,\tD)\in E_{N,M,\varepsilon}^{F}(\rho)$, namely 
$E_{N,M,\varepsilon}^\mathrm{dil}(\rho) \subseteq E_{N,M,\varepsilon}^{F}(\rho)$, whence $I^{F}(\rho)\leq I^\mathrm{dil}(\rho)$.

Now let $(\tE,\tD)\in E^{F}_{N,M,\varepsilon}$, then by definition
\[
F[\Psi,(\tC\boxtimes\tI) (\Psi)]^{2}>1-\varepsilon, \quad \forall \Psi\in D_{\rho^{\boxtimes N}}.
\]
By the second inequality in proposition~\ref{prop:gen_fuchs} we have that
\begin{align*}
&\normop{(\tC\boxtimes\tI)(\Psi)-\Psi}\\
&\leq 2 \sqrt{1-F[\Psi,(\tC\boxtimes\tI) (\Psi)]^{2}}\\
	& <2 \sqrt{\varepsilon},
\end{align*}
for any $\Psi\in D_{\rho^{\boxtimes N}}$, which means that $(\tE,\tD)\in E_{N,M,2\sqrt{\varepsilon}}^\mathrm{dil}(\rho)$,
and the reverse inequality $I^{F}(\rho)\geq I^\mathrm{dil}(\rho)$ follows.
\end{proof}
\section{Info content and states purity}\label{sec:pur}

In the following we will assume that
\begin{align*}
\tE:\rA^{\boxtimes N}\to\rB^{\boxtimes M},\quad\tD:\rB^{\boxtimes M}\to\rA^{\boxtimes N}.
\end{align*}
Moreover, we will denote by $\{\Psi_i\}$ any preparation test of $\rA^{\boxtimes N}\rC$ such that
\begin{align*}
\sum_ie_\rC\circ\Psi_i=\rho^{\boxtimes N}.
\end{align*}
Finally, for every observation-test $\{a_j\}$ of $\rA^{\boxtimes N}\rC$, and for any pair 
$(\tE,\tD)$, let us define the two probability distributions
\begin{align}
p_{i,j}\coloneqq  
\begin{aligned}
   	\Qcircuit @C=1em @R=.7em @! R {\multiprepareC{1}{\Psi_i}&\qw&\ustick{\rA^{\boxtimes N}}\qw&\qw&\multimeasureD{1}{a_j}\\
	\pureghost{\Psi_i}&\qw&\ustick{\rC}\qw&\qw&\ghost{a_j}}
   \end{aligned}\ ,
\end{align}
and
\begin{align}
q_{i,j}\coloneqq
   \begin{aligned}
   	\Qcircuit @C=1em @R=.7em @! R {\multiprepareC{1}{\Psi_i}&\qw&\ustick{\rA^{\boxtimes N}}\qw&\qw&\gate{\tE^N}&\qw&\ustick{\rB^{\boxtimes M}}\qw&\qw&\gate{\tD^N}&\qw&\ustick{\rA^{\boxtimes N}}\qw&\qw&\multimeasureD{1}{a_j}\\
	\pureghost{\Psi_i}&\qw&\qw&\qw&\qw&\qw&\ustick{\rC}\qw&\qw&\qw&\qw&\qw&\qw&\ghost{a_j}}
   \end{aligned}\ .
\end{align}
We can then introduce the following functions that represent the Shannon mutual information between classical random variables
$X$ and $Y$, distributed according to $P(X=x_i,Y=y_j)=p_{i,j}$ or $X$ and $\tilde Y$, distributed according to $P(X=x_i,\tilde Y=y_j)=q_{i,j}$
\begin{align*}
I(X:Y)\coloneqq\sum_{i,j}p_{i,j}\log_2\frac{p_{i,j}}{p^X_ip^Y_j}\\
I(X:\tilde Y)\coloneqq\sum_{i,j}q_{i,j}\log_2\frac{q_{i,j}}{q^X_iq^{\tilde Y}_j},
\end{align*}
where $p^Y_j$, $q^{\tilde Y}_j$ and $q^X_i=p^X_i$ denote the elements of the marginal distributions.

\begin{definition}\label{def:classinfoc}
Let $\{\Psi_i\}$ denote a preparation test such that $\sum_i\Psi_i\in D_{\rho^{\boxtimes N}}$. We denote by 
$E_{N,M,\delta}(\rho)$ the set of those compression schemes such that
\[
\sup_{\rC,\{\Psi_i\},\{a_j\}}L^{-1}|I(X:Y)-I(X,\tilde{Y})|<\delta,
\]
with $L=\log_2(mn-1)$ where $n$ is the cardinality of the preparation test $\{\Psi_i\}$ and $m$ that of the test $\{a_j\}$.
We then define the following quantities
\begin{align}
&R^C_{\delta,N}(\rho)\coloneqq\frac{\min\{M\mid E^C_{N,M,\delta}(\rho)\neq\emptyset\}}N,\\
&R^C_\delta(\rho)\coloneqq\limsup_{N\to\infty}R^C_{\delta,N}(\rho),\\
&I^C(\rho)\coloneqq\lim_{\delta\to0}R^C_\delta(\rho).
\end{align}
\end{definition}
The last quantity above satisfies the following lemmas.
\begin{lemma}\label{lemma:bound_infocontent}
Let $\rho\in\st_1(\rA)$. Then $I(\rho)\geq I^C(\rho)$.
\end{lemma} 
The proof can be found in appendix~\ref{app:purity_lemmas}.

In proving the following lemma and the subsequent proposition we assume that when we compose systems, the size does not increase
more than exponentially. More precisely, we formulate the following assumption that will hold in the remainder.
\begin{assumption}[Regular scaling]
For every type of system $\rA$, there exists a constant $k_\rA>0$ such that 
the size $D(N)\coloneqq D_{\rA^{\boxtimes N}}$ of the compound system $\rA^{\boxtimes N}$ satisfies $D(N)\leq k_\rA D(1)^N$.
\end{assumption}

\begin{lemma}\label{lemma:posbound_classinfocontent}
Let $\rho\in\st_1(\rA)$ be a mixed state, then $I^C(\rho)>0$.
\end{lemma}
The proof can be found in appendix~\ref{app:proofincubo}.

\begin{proposition}\label{prop:null_infocontent}
Let $\rho\in\st_1(\rA)$. If $I(\rho)=0$ then $\rho$ is a pure state.
\end{proposition}
\begin{proof}
Let $\rho$ be mixed. By~\ref{lemma:bound_infocontent} and~\ref{lemma:posbound_classinfocontent}
\[
I(\rho)\geq I^C(\rho)>0,
\]
whence the thesis.
\end{proof}

We now use the above results to prove some general facts about theories with essentially unique purification and 
atomicity of parallel composition for states. For this purpose, let us start reminding the precise meaning of the 
above requirements.

\begin{definition}[Existence of purification]\label{ass:exist_pur}
We say that an OPT satisfies purification if for any $\rho\in\st(\rA)$ one has $P_{\rho}\neq\emptyset$.
\end{definition}

\begin{definition}[Essential uniqueness of purification]\label{ass:unique_pur}
We say that an OPT satisfies essential uniqueness of purification if, for any $\rho\in\st(\rA)$ such that 
$P_{\rho}\neq\emptyset$,  $\forall\Phi,\Psi\in P_{\rho}$ with $\Psi,\Phi\in\st(\rA\rB)$, there exists a reversible transformation $\tU$ such that

\begin{equation}
\begin{aligned}
 \Qcircuit @C=1em @R=.7em @! R { \multiprepareC{1}{\Psi} & \ustick \rA \qw& \qw \\
 \pureghost {\Psi} & \ustick \rB \qw &\qw}
\end{aligned}
\ = \
\begin{aligned}
 \Qcircuit @C=1em @R=.7em @! R { \multiprepareC{1}{\Phi} & \qw& \ustick \rA \qw& \qw \\
 \pureghost {\Phi} & \ustick \rB \qw &\gate{\tU} & \ustick \rB\qw}
\end{aligned}\ \ .
\end{equation}
\end{definition}

\begin{definition}[Atomicity of parallel composition of states]\label{ass:atom_parcomp_st}
We say that an OPT satisfies atomicity of parallel composition of states if for any pair $\phi\in\purst(\rA)$ 
and $\psi\in\purst(\rB)$ we also have $\phi\boxtimes\psi\in\purst(\rA\rB)$.
\end{definition}

We now prove that, in every strongly causal theory that satisfies regular scaling and uniqueness of purification, null 
information content of pure states is equivalent to atomicity of parallel composition of states. We stress that the 
requirement of {\em existence} of purification is not needed, but only its uniqueness. In other  words, {\em if} a 
state has a purification, then it is unique, even though it needs not have one. The following result is particularly 
interesting because it provides an alternative way of understanding the operational content of atomicity of parallel 
composition.

\begin{proposition} Let us consider strongly causal OPT satisfying regular scaling. 
Then the requirements of essential uniqueness of purification and atomicity of parallel composition of states imply 
that $I(\phi)=0$ for any $\phi \in \purst(\rA)$. Conversely, if $I(\phi)=0$ for any $\phi \in \purst(\rA)$, atomicity 
of parallel composition of states holds.
\end{proposition}
\begin{proof}
By lemma~\ref{lemma:pur_infocontent} we have $I(\phi)=I^\mathrm{pur}(\phi)$. Now, let us fix $N$ and consider a dilation $\Omega$
of $\phi^{\boxtimes N}$. Let $\{\Psi_{i}\}_{i\in\sX}$ be a pure decomposition of 
$\Omega$, then by purity of $\phi^{\boxtimes N}$
we must have
\begin{equation}
p_i\begin{aligned}
 \Qcircuit @C=1em @R=.7em @! R { \prepareC{\phi^{\boxtimes N}} & \ustick {\rA^N} \qw & \qw}
\end{aligned}
\ = \
\begin{aligned}
 \Qcircuit @C=1em @R=.7em @! R { \multiprepareC{1}{\Psi_i} & \ustick {\rA^N} \qw& \qw \\
 \pureghost {\Psi_i} & \ustick \rB \qw&\measureD{e}}
\end{aligned}\ , \quad \forall i\in\sX.
\end{equation}
Now, let $\eta\in\purst(\rB)$ and consider $\phi^{\boxtimes N}\boxtimes \eta$. This is still a pure state,
hence a purification of $\phi^{\boxtimes N}$. Therefore, by essential uniqueness of purifications

\begin{equation}
\begin{aligned}
 \Qcircuit @C=1.2em @R=.7em @! R { \multiprepareC{1}{\Psi_{i}} & \ustick {\rA^{\boxtimes N}} \qw& \qw \\
 \pureghost {\Psi_{i}} & \ustick \rB \qw &\qw}
\end{aligned}
\ = \ p_i
\begin{aligned}
 \Qcircuit @C=1em @R=.7em @! R { \prepareC{\phi^{\boxtimes N}} & \qw& \ustick {\rA^{\boxtimes N}} \qw& \qw \\
\prepareC{\eta} & \ustick \rB \qw &\gate{\tU_{i}} & \ustick \rB\qw}
\end{aligned}\ \ , \quad \forall i\in\sX.
\end{equation}
where $\tU_{i}$ are reversible channels on $\rB$. 

Now let us consider a compression scheme defined by a measure and prepare one, as follows
\begin{equation}
\tE\coloneqq
\begin{aligned}
\Qcircuit @C=1em @R=.7em @! R {& \ustick {\rA^{\boxtimes N}}\qw &\measureD {e}}
\end{aligned}\ , \quad
\tD\coloneqq
\begin{aligned}
\Qcircuit @C=1em @R=.7em @! R {&\prepareC {\phi^{\boxtimes N}}& \qw& \ustick {\rA^{\boxtimes N}}\qw  }
\end{aligned}\ .
\end{equation}
It is clear that for any dilation $\Omega$ of $\phi^{\boxtimes N}$, the above scheme is such that
$(\tD\tE\boxtimes\tI)(\Psi_i)-\Psi_i = 0$ and this implies that for any $N$ and $\varepsilon$ we have 
$E_{N,0,\varepsilon}^\mathrm{pur}(\rho)\neq\emptyset$ and then $I(\rho)=I^\mathrm{pur}(\rho)=0$.

Now, let us assume that for any A and for any $\phi\in\purst(\rA)$ we have $I(\phi)=0$. Let $\rho\in\purst(\rA)$ and
$\sigma\in\purst(\rB)$. By proposition~\ref{prop:subadditivity} we have that
\[
I(\rho\boxtimes\sigma)\leq I(\rho)+I(\sigma)=0.
\]
Thus $I(\rho\boxtimes\sigma)=0$, and by proposition~\ref{prop:null_infocontent} $\rho\boxtimes\sigma$ is pure, namely,
assumption~\ref{ass:atom_parcomp_st} holds.
\end{proof}
It is interesting to observe that, thanks to the above proposition, one can exhibit operational probabilistic theories having pure states with non-vanishing information content. In Ref.~\cite{PhysRevA.102.052216} the authors construct a theory, named \emph{bilocal classical theory}, where all systems are classical (the set of states is a simplex), but with a parallel composition rule that differs from the one of classical information theory, thereby violating assumption~\ref{ass:atom_parcomp_st}. Accordingly, bilocal classical theory must have pure states with non-null information content.

\section{Information content in Quantum and Classical Information Theory}\label{sec:boilsdown}

Before restricting to the quantum case, we prove the following lemma concerning
the correlation fidelity defined in~\ref{def:ent_fid}.

\begin{lemma}
Let $\rho\in\st_1(\rA)$ and $\tC\in\tr_{1}(\rA)$. If every state has a purification (definition~\ref{ass:exist_pur}), then one has
\begin{equation}
F(\rho,\tC)=\inf_{\Phi\in P_{\rho}}F[\Phi,\tC\boxtimes\tI(\Phi)]^{2}.
\end{equation}
Moreover, in an OPT with essential uniqueness of purification (def.~\ref{ass:atom_parcomp_st}) and atomicity of parallel composition of states 
(def.~\ref{ass:unique_pur}), for any $\Phi\in P_{\rho}$ one has
\begin{equation}
F(\rho,\tC)=F[\Phi,\tC\boxtimes\tI(\Phi)]^{2}.
\end{equation}
\end{lemma}
\begin{proof}
If $P_{\tau}\neq\emptyset$ for every state $\tau$, then for any $\Psi\in D_{\rho}$ one has that there exists
$\Gamma\in P_{\Psi}\subseteq P_{\rho}$. Therefore, by monotonicity of the fidelity (prop.~\ref{prop:fid_monot})
we have
\begin{align*}
F[\Psi,\tC\boxtimes\tI(\Psi)]^{2} & \geq F[\Gamma,\tC\boxtimes\tI(\Gamma)]^{2}\geq \\
						& \geq \inf_{\Phi\in P_{\rho}}F[\Phi,\tC\boxtimes\tI(\Phi)]^{2}.
\end{align*}
Since this holds for any $\Psi\in D_{\rho}$, it implies $F(\rho,\tC)\geq \inf_{\Phi\in P_{\rho}}F[\Phi,\tC\boxtimes\tI(\Phi)]^{2}$. 
The reverse inequality is trivial, since $P_{\rho}\subseteq D_{\rho}$.

If all the purifications of $\rho$ with the same purifying system are connected through a reversible transformation $\tU$
and atomicity of parallel composition of pure states also hold (see def.~\ref{ass:atom_parcomp_st} and~\ref{ass:unique_pur}), then, for any fixed purification $\Phi$ in $P_\rho$, and any other
$\Gamma$ in $P_{\rho}$, there exists a channel $\tA\in\tr_{1}(\rB,\rC)$ such that

\begin{equation}
\begin{aligned}
 \Qcircuit @C=1em @R=.7em @! R { \multiprepareC{1}{\Gamma} & \ustick \rA \qw& \qw \\
 \pureghost {\Gamma} & \ustick \rC \qw &\qw}
\end{aligned}
\ = \
\begin{aligned}
 \Qcircuit @C=1em @R=.7em @! R { \multiprepareC{1}{\Phi} & \qw& \ustick \rA \qw& \qw \\
 \pureghost {\Phi} & \ustick \rB \qw &\gate{\tA} & \ustick \rC\qw}
\end{aligned}\ \ .
\end{equation}
By monotonicity one has $F(\rho,\tC)\geq F[\Phi,\tC\boxtimes\tI(\Phi)]^{2}$ and the reverse inequality is trivial, as $\Phi \in P_{\rho}$.
\end{proof}

We recall the statement of the Schumacher theorem. Let $\rho\in\st_{1}(\rho)$ and $(\tE,\tD)$ be a compression scheme
\begin{theorem}[Schumacher]
Let $\rho\in\st_1(\rA)$ with $\cH_{\rA}$ the Hilbert space corresponding to the quantum system $\rA$,
let $(\tE,\tD)$ be a compression scheme and define its ratio $R$ as
\[
R\coloneqq\frac{\log[{\rm dim}[\Supp(\tE(\rho^{\otimes N}))]]}{N}.
\]
For every $\varepsilon > 0 $ and $R>S(\rho)$ there exists $N_0$ such that $\forall N\geq N_0$ there exists a
compression scheme with ratio $R$ such that $F(\rho^{\otimes N},\tD\tE)>1-\varepsilon$.
Conversely, for every $R<S(\rho)$ there is $\varepsilon>0$ such that  for every compression scheme $(\tE,\tD)$
with ratio $R$ one has $F(\rho^{\otimes N},\tD\tE)\leq\varepsilon$.
\end{theorem}

\begin{proposition}\label{prop:quantum_infocontent}
Let $\rho\in\st_1(\rA)$ be a quantum state and denote with $S(\rho)$ its von Neumann entropy. Then $I(\rho)=S(\rho)$.
\end{proposition}
\begin{proof}
We start by showing that $I^F(\rho)\leq S(\rho)$. Let $\delta>0$, $R\in (S(\rho),S(\rho)+\delta]$ and $\varepsilon>0$.
By the direct part of the Schumacher theorem there exists a $N_{0}$ such that for any $N\geq N_0$ there is a 
$(N,\varepsilon)$-reliable compression scheme with rate $R$. Thus, upon embedding $\Supp[\tE(\rho^{\otimes N})]$
in $\ceil{NR}$ qubits, by using a suitable isometry, we have $E^{F}_{N,\ceil{NR},\varepsilon}(\rho)\neq \emptyset$ for any
$N\geq N_0$. This implies
\begin{align*}
&\limsup_{N\rightarrow \infty}\frac{\min\{M:E^{F}_{N,M,\varepsilon}(\rho)\neq \emptyset\}}{N} \\
		\leq& \lim_{N\rightarrow \infty}\frac{\ceil{NR}}{N} = R\leq S(\rho)+\delta.
\end{align*}
Since the argument holds for any $\varepsilon >0$, we get $I^F(\rho)\leq S(\rho)+\delta$, and being delta arbitrary, we find
$I^F(\rho)\leq S(\rho)$.

Now let $\delta > 0$ and consider $M/\overline{N}$ such that $S(\rho)-\delta\leq M/\overline{N}<S(\rho)$. By the converse part
of Schumacher theorem there exists $\overline{\varepsilon}>0$ such that for any compression scheme with ratio 
$M/\overline{N}$ one has $F(\rho^{\otimes N},\tD\tE)\leq 1- \overline{\varepsilon}$. In particular, since any compression
scheme from $k\overline{N}$ copies of the system to $kM$ qubits has ratio $M/\overline{N}$, one has
\[
E_{k\overline{N},kM,\overline{\varepsilon}}^{F}(\rho)=\emptyset,\quad \forall k \in \N.
\]
Therefore, for any $0<\varepsilon<\overline{\varepsilon}$ and $k$
\begin{align*}
S(\rho)-\delta\leq\frac{M}{\overline{N}}&<\frac{\min \{L:E^F_{k\overline{N},L,\overline{\varepsilon}}(\rho)\neq\emptyset\}}{k\overline{N}}\leq \\
&\leq \frac{\min \{L:E^F_{k\overline{N},L,\varepsilon}(\rho)\neq\emptyset\}}{k\overline{N}}.
\end{align*}
Thus, by taking the $\limsup_{k\rightarrow \infty}$ we find that
\begin{align*}
S(\rho)-\delta &\leq \limsup_{k\rightarrow \infty}\frac{\min \{L:E^F_{k\overline{N},L,\varepsilon}(\rho)\neq\emptyset\}}{k\overline{N}}\\
	&\leq \limsup_{N\rightarrow \infty}\frac{\min \{L:E^F_{N,L,\varepsilon}(\rho)\neq\emptyset\}}{N},
\end{align*}
for any $0<\varepsilon<\overline{\varepsilon}$. By taking the $\lim_{\varepsilon\rightarrow 0 }$ and the arbitrariness
of $\delta$ we finally get $S(\rho)\leq I^F(\rho)$. The statement then follows by the fact that in quantum theory one has $I^F(\rho)=I(\rho)$ 
(prop.~\ref{prop:infocontent_dil} and~\ref{prop:infocontent_fidelity}).
\end{proof}

Let us now turn our focus to the classical case. In this setting, the input and the output of the compression scheme
are given by strings of $N$ letters drawn from an alphabet $\chi$. Each letter $x_{i}$ appears with a given 
probability 
$p_{i}$ and the probability that the overall string $x_{i_{1}}\dots x_{i_{N}}$ is emitted is given by the joint probability 
$p_{i_{1},\dots,i_{N}}$. If we assume that each symbol is independently and identically distributed, then 
$p_{i_{1},\dots,i_{N}}=p_{i_{1}}\dots p_{i_{N}}$. The probability $p(e)$
of emitting an output string which is different form the input one is often considered in the literature as a figure 
of merit.  More formally, let $C$ be the 
Markov matrix representing the composition
of the compression and decompression maps, and ${\bf i}\coloneqq i_{1}\dots i_{N}$ define the input string,
then the error probability is defined as 
\[
p^{C}(e)\coloneqq\sum_{\bi}\sum_{\bj\neq\bi}p(\bi\neq\bj|\bi)=1-\sum_{\bi}C_{\bi,\bi}p_{\bi}.
\]
The set of states of classical theory is given by a simplex, and any probability vector representing a state can be uniquely
decomposed in terms of pure states $\be_{i}$, corresponding to vectors with all zero components except the one
in the $i$-th position: $(e_i)_j=\delta_{i,j}$. Since $C$ is a stochastic matrix we have the following chain of 
equalities
\begin{align*}
&\sum_{\bi}p_{\bi}\norma{C\be_{\bi}-\be_{\bi}}\\
&=\sum_{\bi}p_{\bi}\sum_{\bj}|C_{\bj,\bi}-\delta_{\bi\bj}|  \\
&=\sum_{\bi}p_{\bi}(\sum_{\bj\neq\bi}C_{\bj,\bi}+1-C_{\bi,\bi}) \\
&=\sum_{\bi}2p_{\bi}(1-C_{\bi,\bi}) = \\
&=2\sum_{\bi}p_{\bi}(1-C_{\bi,\bi}) = 2(1-\sum_{\bi}p_\bi C_{\bi,\bi})=2p^{C}(e),
\end{align*}
namely
\[
p^C(e)=\sum_{\bi}p_{\bi}\frac{1}{2}\norma{C\be_{\bi}-\be_{\bi}}_{1}.
\]
Now consider \emph{the unique} pure decomposition of some dilation $\bPi\in D_{\bp^{\otimes N}}$. This is given by $\bPi=\sum_{\bi,j}\Pi_{ij}\be_{\bi}\otimes \be_{j}$ with $\sum_j \Pi_{\bi j}=p_{\bi}$(the pure states of the composite system are the tensor product vectors
of the pure ones of the composing systems). Then we find 
\begin{align*}
&\sum_{\bi,j}\Pi_{\bi j}\norma{(C\otimes I)\be_{\bi}\otimes\be_j - \be_{\bi}\otimes\be_j}_{1} \\
&=\sum_{\bi,j}\Pi_{\bi j}\norma{(C - I)\be_{\bi}\otimes\be_j }_{1} \\
&=\sum_{\bi,j}\Pi_{\bi j}\norma{(C - I)\be_{\bi} }_{1} \\
&=\sum_{\bi}p_{\bi}\norma{(C - I)\be_{\bi}}_{1} =2 p^C(e),
\end{align*}
having used the fact that for any $j$ 
\[
\norma{(C - I)\be_{\bi}\otimes\be_j }_{1}=\norma{(C - I)\be_{\bi} }_{1}.
\]
Summarizing, we have proved that for any dilation $\bPi\in D_{\bp^{\otimes N}}$
\[
p^{C}(e)=\frac{1}{2}\sum_{\bi,j}\Pi_{\bi j}\norma{(C\otimes I)\be_{\bi}\otimes\be_j - \be_{\bi}\otimes\be_j}_{1}
\]
Namely, in the classical case, the error probability is exactly our figure of merit in~\eqref{eq:pure_figure_merit}. 

Now, one can use the first Shannon theorem in order to prove that the information content of a classical state is exactly
its Shannon entropy

\begin{theorem}[{Shannon}]
Let $\bp\in\st_1(\rA)$ be a classical state,
let $(\tE,\tD)$ be a compression scheme and define its ratio $R$ as
\[
R\coloneqq\frac{\log_{2}|\tE(\Rng ({\rm X}^{N}))|}{N}.
\]
For every $\varepsilon > 0 $ and $R>H(X)$ there exists $N_0$ such that $\forall N\geq N_0$ there exists a
compression scheme with ratio $R$ such that $p^{C}(e)<\varepsilon$.
Conversely, for every $R<H(X)$ there is $\varepsilon>0$ such that  for every compression scheme $(\tE,\tD)$
with ratio $R$ one has $p^{C}(e)\geq\varepsilon$.
\end{theorem} 

A propoisition analogous to~\ref{prop:quantum_infocontent} holds for the classical case, whose proof is essentially the same.
The main issue in both cases is to correctly identify the figure of merit that must be adopted in order to define the information content.

\section{Conclusion}\label{sec:concl}
We have defined the information content for a source of information of an arbitrary operational probabilistic theory. 
The only assumption needed is that of  digitalizability: a theory is digitalizable if any system of the theory can be perfectly 
mapped into finitely many copies of a reference system, called ``obit'', playing the role that ``bit'' and ``qubit'' play in classical and quantum theory, respectively. The information content of a source is defined as the minimum number of obits needed to store the output of 
the source in a such a way that it can be recovered with arbitrary accuracy. The figure of merit for establishing accuracy, 
independently of the features of the theory, is robust against any distortion effect that a compression scheme could induce on the 
state of the source, on its admissible preparations and on the correlations with external systems. Accordingly, the figure of merit 
meets the following two criteria: i) any preparation of ensembles that average to the considered state must be indistinguishable 
from leaving the preparation untouched, ii) the compression scheme must preserve decompositions of dilations of the state 
of interest, namely joint states of the system and arbitrary external systems such that the state that one obtains after averaging and discarding the external system is  the one at hand. 

We first proved that the information content is always a well defined quantity. Moreover, in the hypothesis of steering of 
ensembles, we show that the information content can be computed using simple figures of merit, e.g.~a generalisation of entanglement fidelity here denoted by correlation fidelity. Then we show that the present notion of 
information content coincides with the Shannon and von Neumann entropies in the classical and quantum case, respectively. 
In quantum theory both the entanglement fidelity and the average input-output fidelity for an arbitrary decomposition of the state representing the source identify the von Neumann entropy as the minimal compression rate. This opens a relevant question: 
which are the minimal assumptions behind the collapse of ``global'' and ``local'' figures of merit, namely quantifiers of the ability to recover the correlations of the source, and the local preparations of the source, respectively?


Like Shannon's and von Neumann's entropy, we proved that the information content is subadditive, and can be used to measure the purity of a state. Indeed both Shannon's and von Neumann's entropy vanish if and only if the state is pure, and here we show that the information content has this feature as well. While it is always true that a source with null information content corresponds to a pure state, the opposite implication is satisfied in the presence of atomicity of parallel composition (the parallel composition of any two pure states is pure) and unique purification (if a state has a purification, then the latter is unique up to reversible channels on the remote system).

In the light of the above results we propose the information content as a candidate entropic quantity generalising Shannon entropy of classical systems and von Neumann entropy of quantum systems. A basic message across the literature on general probabilistic theories~\cite{Barnum:2010gy,Short:2010kt,2010RpMP...66..175K} is that a theory is usually not \emph{monoentropic}, namely multiple entropic quantities can be defined, each one reducing to Shannon's and von Neumann's entropy in classical and quantum theory, respectively. The notions of entropy usually considered are defined in terms of classical information quantities as follows. i) The \emph{measurement entropy} of a system, namely the infimum Shannon entropy
of any possible measurement on the system, measures the minimum
uncertainty of measurement among indecomposable measurements under a state. 
ii) The \emph{decomposition (or mixing)
entropy}, namely the infimum of the Shannon entropies over
all possible ways of preparing the system's state as a mixture
of pure states, measures the minimum uncertainty for a preparation of a state with respect to pure states. iii) The supremum of the Shannon mutual information between two random variables related respectively to preparations and measurements on the system, measures the maximum \emph{accessible information}. 

For a general theory the above quantities can be very different and violate some of the typical features of Shannon and von Neumann entropies, such as \emph{concavity} and \emph{strong subadditivity}.  It is known \cite{Barnum:2010gy,Short:2010kt,2010RpMP...66..175K} that measurement entropy is both subadditive and concave but in general it is not strongly subadditive and does not provide a measure of purity, while decomposition entropy is generally neither subadditive nor concave. For example, in Ref.~\cite{Barnum:2010gy} it is shown that, for systems whose state space is a non-simplicial polytope, the decomposition entropy is not concave. Less is known about the third entropic quantity given in terms of the Shannon mutual information, which still could satisfy all features of Shannon and von Neumann entropies. 

One of the main outcome of this manuscript are a series of results that can be used to clarify which of the possible entropies of a general probabilistic theory has operational meaning in terms of optimal compression ratio. We started here the analysis of the relation between information content and other entropies of a general probabilistic theory focusing on the accessible information. On one hand, both quantities provide a measure of purity of a state, and on the other hand we proved that the accessible information is a lower bound for information content. 
An important open question is under what conditions the information content, which by definition is the optimal compression ratio, coincides with the accessible information in general.

\acknowledgments
A. T. acknowledges financial support from Elvia and Federico Faggin foundation through Silicon Valley Community Foundation, Grant No. 2020-214365.

\bibliography{biblio-ic}

\begin{thebibliography}{28}%
\makeatletter
\providecommand \@ifxundefined [1]{%
 \@ifx{#1\undefined}
}%
\providecommand \@ifnum [1]{%
 \ifnum #1\expandafter \@firstoftwo
 \else \expandafter \@secondoftwo
 \fi
}%
\providecommand \@ifx [1]{%
 \ifx #1\expandafter \@firstoftwo
 \else \expandafter \@secondoftwo
 \fi
}%
\providecommand \natexlab [1]{#1}%
\providecommand \enquote  [1]{``#1''}%
\providecommand \bibnamefont  [1]{#1}%
\providecommand \bibfnamefont [1]{#1}%
\providecommand \citenamefont [1]{#1}%
\providecommand \href@noop [0]{\@secondoftwo}%
\providecommand \href [0]{\begingroup \@sanitize@url \@href}%
\providecommand \@href[1]{\@@startlink{#1}\@@href}%
\providecommand \@@href[1]{\endgroup#1\@@endlink}%
\providecommand \@sanitize@url [0]{\catcode `\\12\catcode `\$12\catcode
  `\&12\catcode `\#12\catcode `\^12\catcode `\_12\catcode `\%12\relax}%
\providecommand \@@startlink[1]{}%
\providecommand \@@endlink[0]{}%
\providecommand \url  [0]{\begingroup\@sanitize@url \@url }%
\providecommand \@url [1]{\endgroup\@href {#1}{\urlprefix }}%
\providecommand \urlprefix  [0]{URL }%
\providecommand \Eprint [0]{\href }%
\providecommand \doibase [0]{http://dx.doi.org/}%
\providecommand \selectlanguage [0]{\@gobble}%
\providecommand \bibinfo  [0]{\@secondoftwo}%
\providecommand \bibfield  [0]{\@secondoftwo}%
\providecommand \translation [1]{[#1]}%
\providecommand \BibitemOpen [0]{}%
\providecommand \bibitemStop [0]{}%
\providecommand \bibitemNoStop [0]{.\EOS\space}%
\providecommand \EOS [0]{\spacefactor3000\relax}%
\providecommand \BibitemShut  [1]{\csname bibitem#1\endcsname}%
\let\auto@bib@innerbib\@empty
\bibitem [{\citenamefont {Shannon}(1949)}]{shannon1949communication}%
  \BibitemOpen
  \bibfield  {author} {\bibinfo {author} {\bibfnamefont {C.~E.}\ \bibnamefont
  {Shannon}},\ }\href@noop {} {\bibfield  {journal} {\bibinfo  {journal}
  {Proceedings of the IRE}\ }\textbf {\bibinfo {volume} {37}},\ \bibinfo
  {pages} {10} (\bibinfo {year} {1949})}\BibitemShut {NoStop}%
\bibitem [{\citenamefont {Schumacher}(1995)}]{PhysRevA.51.2738}%
  \BibitemOpen
  \bibfield  {author} {\bibinfo {author} {\bibfnamefont {B.}~\bibnamefont
  {Schumacher}},\ }\href {\doibase 10.1103/PhysRevA.51.2738} {\bibfield
  {journal} {\bibinfo  {journal} {Phys. Rev. A}\ }\textbf {\bibinfo {volume}
  {51}},\ \bibinfo {pages} {2738} (\bibinfo {year} {1995})}\BibitemShut
  {NoStop}%
\bibitem [{\citenamefont {Koashi}\ and\ \citenamefont
  {Imoto}(2001)}]{Koashi2001}%
  \BibitemOpen
  \bibfield  {author} {\bibinfo {author} {\bibfnamefont {M.}~\bibnamefont
  {Koashi}}\ and\ \bibinfo {author} {\bibfnamefont {N.}~\bibnamefont {Imoto}},\
  }\href {\doibase 10.1103/physrevlett.87.017902} {\bibfield  {journal}
  {\bibinfo  {journal} {Phys. Rev. Lett.}\ }\textbf {\bibinfo {volume} {87}},\
  \bibinfo {pages} {017902} (\bibinfo {year} {2001})}\BibitemShut {NoStop}%
\bibitem [{\citenamefont {Barnum}\ \emph {et~al.}(2001)\citenamefont {Barnum},
  \citenamefont {Caves}, \citenamefont {Fuchs}, \citenamefont {Jozsa},\ and\
  \citenamefont {Schumacher}}]{Barnum_2001}%
  \BibitemOpen
  \bibfield  {author} {\bibinfo {author} {\bibfnamefont {H.}~\bibnamefont
  {Barnum}}, \bibinfo {author} {\bibfnamefont {C.~M.}\ \bibnamefont {Caves}},
  \bibinfo {author} {\bibfnamefont {C.~A.}\ \bibnamefont {Fuchs}}, \bibinfo
  {author} {\bibfnamefont {R.}~\bibnamefont {Jozsa}}, \ and\ \bibinfo {author}
  {\bibfnamefont {B.}~\bibnamefont {Schumacher}},\ }\href {\doibase
  10.1088/0305-4470/34/35/304} {\bibfield  {journal} {\bibinfo  {journal} {J.
  Phys. A: Math. Gen. 34 6767}\ }\textbf {\bibinfo {volume} {34}},\ \bibinfo
  {pages} {6767} (\bibinfo {year} {2001})}\BibitemShut {NoStop}%
\bibitem [{\citenamefont {Khanian}\ and\ \citenamefont
  {Winter}(2019)}]{Winter_2019}%
  \BibitemOpen
  \bibfield  {author} {\bibinfo {author} {\bibfnamefont {Z.~B.}\ \bibnamefont
  {Khanian}}\ and\ \bibinfo {author} {\bibfnamefont {A.}~\bibnamefont
  {Winter}},\ }in\ \href@noop {} {\emph {\bibinfo {booktitle} {2019 IEEE
  International Symposium on Information Theory (ISIT)}}}\ (\bibinfo {year}
  {2019})\ pp.\ \bibinfo {pages} {1147--1151}\BibitemShut {NoStop}%
\bibitem [{\citenamefont {Schumacher}\ and\ \citenamefont
  {Westmoreland}(2010)}]{schumacher_westmoreland_2010}%
  \BibitemOpen
  \bibfield  {author} {\bibinfo {author} {\bibfnamefont {B.}~\bibnamefont
  {Schumacher}}\ and\ \bibinfo {author} {\bibfnamefont {M.}~\bibnamefont
  {Westmoreland}},\ }\href {\doibase 10.1017/CBO9780511814006} {\emph {\bibinfo
  {title} {Quantum Processes Systems, and Information}}}\ (\bibinfo
  {publisher} {Cambridge University Press},\ \bibinfo {year}
  {2010})\BibitemShut {NoStop}%
\bibitem [{\citenamefont {Hardy}\ and\ \citenamefont
  {Wootters}(2012)}]{hardy2012limited}%
  \BibitemOpen
  \bibfield  {author} {\bibinfo {author} {\bibfnamefont {L.}~\bibnamefont
  {Hardy}}\ and\ \bibinfo {author} {\bibfnamefont {W.~K.}\ \bibnamefont
  {Wootters}},\ }\href {\doibase 10.1007/s10701-011-9616-6} {\bibfield
  {journal} {\bibinfo  {journal} {Foundations of Physics}\ }\textbf {\bibinfo
  {volume} {42}},\ \bibinfo {pages} {454} (\bibinfo {year} {2012})}\BibitemShut
  {NoStop}%
\bibitem [{\citenamefont {D’Ariano}\ \emph
  {et~al.}(2014{\natexlab{a}})\citenamefont {D’Ariano}, \citenamefont
  {Manessi}, \citenamefont {Perinotti},\ and\ \citenamefont
  {Tosini}}]{2014fermionic}%
  \BibitemOpen
  \bibfield  {author} {\bibinfo {author} {\bibfnamefont {G.~M.}\ \bibnamefont
  {D’Ariano}}, \bibinfo {author} {\bibfnamefont {F.}~\bibnamefont {Manessi}},
  \bibinfo {author} {\bibfnamefont {P.}~\bibnamefont {Perinotti}}, \ and\
  \bibinfo {author} {\bibfnamefont {A.}~\bibnamefont {Tosini}},\ }\href
  {\doibase 10.1209/0295-5075/107/20009} {\bibfield  {journal} {\bibinfo
  {journal} {EPL (Europhysics Letters)}\ }\textbf {\bibinfo {volume} {107}},\
  \bibinfo {pages} {20009} (\bibinfo {year} {2014}{\natexlab{a}})}\BibitemShut
  {NoStop}%
\bibitem [{\citenamefont {D’Ariano}\ \emph
  {et~al.}(2014{\natexlab{b}})\citenamefont {D’Ariano}, \citenamefont
  {Manessi}, \citenamefont {Perinotti},\ and\ \citenamefont
  {Tosini}}]{doi:10.1142/S0217751X14300257}%
  \BibitemOpen
  \bibfield  {author} {\bibinfo {author} {\bibfnamefont {G.~M.}\ \bibnamefont
  {D’Ariano}}, \bibinfo {author} {\bibfnamefont {F.}~\bibnamefont {Manessi}},
  \bibinfo {author} {\bibfnamefont {P.}~\bibnamefont {Perinotti}}, \ and\
  \bibinfo {author} {\bibfnamefont {A.}~\bibnamefont {Tosini}},\ }\href
  {\doibase 10.1142/s0217751x14300257} {\bibfield  {journal} {\bibinfo
  {journal} {Int. J. Mod. Phys. A}\ }\textbf {\bibinfo {volume} {29}},\
  \bibinfo {pages} {1430025} (\bibinfo {year}
  {2014}{\natexlab{b}})}\BibitemShut {NoStop}%
\bibitem [{\citenamefont {D'Ariano}\ \emph {et~al.}(2017)\citenamefont
  {D'Ariano}, \citenamefont {Chiribella},\ and\ \citenamefont
  {Perinotti}}]{DAriano:2017aa}%
  \BibitemOpen
  \bibfield  {author} {\bibinfo {author} {\bibfnamefont {G.~M.}\ \bibnamefont
  {D'Ariano}}, \bibinfo {author} {\bibfnamefont {G.}~\bibnamefont
  {Chiribella}}, \ and\ \bibinfo {author} {\bibfnamefont {P.}~\bibnamefont
  {Perinotti}},\ }\href@noop {} {\emph {\bibinfo {title} {Quantum theory from
  first principles: an informational approach}}}\ (\bibinfo  {publisher}
  {Cambridge University Press},\ \bibinfo {year} {2017})\BibitemShut {NoStop}%
\bibitem [{\citenamefont {Chiribella}\ \emph {et~al.}(2010)\citenamefont
  {Chiribella}, \citenamefont {D'Ariano},\ and\ \citenamefont
  {Perinotti}}]{PhysRevA.81.062348}%
  \BibitemOpen
  \bibfield  {author} {\bibinfo {author} {\bibfnamefont {G.}~\bibnamefont
  {Chiribella}}, \bibinfo {author} {\bibfnamefont {G.~M.}\ \bibnamefont
  {D'Ariano}}, \ and\ \bibinfo {author} {\bibfnamefont {P.}~\bibnamefont
  {Perinotti}},\ }\href {\doibase 10.1103/PhysRevA.81.062348} {\bibfield
  {journal} {\bibinfo  {journal} {Phys. Rev. A}\ }\textbf {\bibinfo {volume}
  {81}},\ \bibinfo {pages} {062348} (\bibinfo {year} {2010})}\BibitemShut
  {NoStop}%
\bibitem [{\citenamefont {Chiribella}\ \emph {et~al.}(2011)\citenamefont
  {Chiribella}, \citenamefont {D'Ariano},\ and\ \citenamefont
  {Perinotti}}]{PhysRevA.84.012311}%
  \BibitemOpen
  \bibfield  {author} {\bibinfo {author} {\bibfnamefont {G.}~\bibnamefont
  {Chiribella}}, \bibinfo {author} {\bibfnamefont {G.~M.}\ \bibnamefont
  {D'Ariano}}, \ and\ \bibinfo {author} {\bibfnamefont {P.}~\bibnamefont
  {Perinotti}},\ }\href {\doibase 10.1103/PhysRevA.84.012311} {\bibfield
  {journal} {\bibinfo  {journal} {Phys. Rev. A}\ }\textbf {\bibinfo {volume}
  {84}},\ \bibinfo {pages} {012311} (\bibinfo {year} {2011})}\BibitemShut
  {NoStop}%
\bibitem [{\citenamefont {D'Ariano}\ \emph
  {et~al.}(2020{\natexlab{a}})\citenamefont {D'Ariano}, \citenamefont {Erba},\
  and\ \citenamefont {Perinotti}}]{PhysRevA.101.042118}%
  \BibitemOpen
  \bibfield  {author} {\bibinfo {author} {\bibfnamefont {G.~M.}\ \bibnamefont
  {D'Ariano}}, \bibinfo {author} {\bibfnamefont {M.}~\bibnamefont {Erba}}, \
  and\ \bibinfo {author} {\bibfnamefont {P.}~\bibnamefont {Perinotti}},\ }\href
  {\doibase 10.1103/PhysRevA.101.042118} {\bibfield  {journal} {\bibinfo
  {journal} {Phys. Rev. A}\ }\textbf {\bibinfo {volume} {101}},\ \bibinfo
  {pages} {042118} (\bibinfo {year} {2020}{\natexlab{a}})}\BibitemShut
  {NoStop}%
\bibitem [{\citenamefont {Hardy}(1999)}]{hardy1999disentangling}%
  \BibitemOpen
  \bibfield  {author} {\bibinfo {author} {\bibfnamefont {L.}~\bibnamefont
  {Hardy}},\ }\href@noop {} {\enquote {\bibinfo {title} {Disentangling
  nonlocality and teleportation},}\ } (\bibinfo {year} {1999}),\ \Eprint
  {http://arxiv.org/abs/quant-ph/9906123} {arXiv:quant-ph/9906123 [quant-ph]}
  \BibitemShut {NoStop}%
\bibitem [{\citenamefont {Spekkens}(2007)}]{Spekkens:2007aa}%
  \BibitemOpen
  \bibfield  {author} {\bibinfo {author} {\bibfnamefont {R.~W.}\ \bibnamefont
  {Spekkens}},\ }\href {\doibase 10.1103/PhysRevA.75.032110} {\bibfield
  {journal} {\bibinfo  {journal} {Phys. Rev. A}\ }\textbf {\bibinfo {volume}
  {75}},\ \bibinfo {pages} {032110} (\bibinfo {year} {2007})}\BibitemShut
  {NoStop}%
\bibitem [{\citenamefont {Barrett}(2007)}]{PhysRevA.75.032304}%
  \BibitemOpen
  \bibfield  {author} {\bibinfo {author} {\bibfnamefont {J.}~\bibnamefont
  {Barrett}},\ }\href {\doibase 10.1103/PhysRevA.75.032304} {\bibfield
  {journal} {\bibinfo  {journal} {Phys. Rev. A}\ }\textbf {\bibinfo {volume}
  {75}},\ \bibinfo {pages} {032304} (\bibinfo {year} {2007})}\BibitemShut
  {NoStop}%
\bibitem [{\citenamefont {Barnum}\ \emph {et~al.}(2010)\citenamefont {Barnum},
  \citenamefont {Barrett}, \citenamefont {Orloff~Clark}, \citenamefont
  {Leifer}, \citenamefont {Spekkens}, \citenamefont {Stepanik}, \citenamefont
  {Wilce},\ and\ \citenamefont {Wilke}}]{Barnum:2010gy}%
  \BibitemOpen
  \bibfield  {author} {\bibinfo {author} {\bibfnamefont {H.}~\bibnamefont
  {Barnum}}, \bibinfo {author} {\bibfnamefont {J.}~\bibnamefont {Barrett}},
  \bibinfo {author} {\bibfnamefont {L.}~\bibnamefont {Orloff~Clark}}, \bibinfo
  {author} {\bibfnamefont {M.}~\bibnamefont {Leifer}}, \bibinfo {author}
  {\bibfnamefont {R.}~\bibnamefont {Spekkens}}, \bibinfo {author}
  {\bibfnamefont {N.}~\bibnamefont {Stepanik}}, \bibinfo {author}
  {\bibfnamefont {A.}~\bibnamefont {Wilce}}, \ and\ \bibinfo {author}
  {\bibfnamefont {R.}~\bibnamefont {Wilke}},\ }\href {\doibase
  10.1088/1367-2630/12/3/033024} {\bibfield  {journal} {\bibinfo  {journal}
  {New Journal of Physics}\ }\textbf {\bibinfo {volume} {12}},\ \bibinfo
  {pages} {3024} (\bibinfo {year} {2010})}\BibitemShut {NoStop}%
\bibitem [{\citenamefont {Short}\ and\ \citenamefont
  {Wehner}(2010)}]{Short:2010kt}%
  \BibitemOpen
  \bibfield  {author} {\bibinfo {author} {\bibfnamefont {A.~J.}\ \bibnamefont
  {Short}}\ and\ \bibinfo {author} {\bibfnamefont {S.}~\bibnamefont {Wehner}},\
  }\href {\doibase 10.1088/1367-2630/12/3/033023} {\bibfield  {journal}
  {\bibinfo  {journal} {New Journal of Physics}\ }\textbf {\bibinfo {volume}
  {12}},\ \bibinfo {pages} {3023} (\bibinfo {year} {2010})}\BibitemShut
  {NoStop}%
\bibitem [{\citenamefont {Kimura}\ \emph {et~al.}(2010)\citenamefont {Kimura},
  \citenamefont {Nuida},\ and\ \citenamefont {Imai}}]{2010RpMP...66..175K}%
  \BibitemOpen
  \bibfield  {author} {\bibinfo {author} {\bibfnamefont {G.}~\bibnamefont
  {Kimura}}, \bibinfo {author} {\bibfnamefont {K.}~\bibnamefont {Nuida}}, \
  and\ \bibinfo {author} {\bibfnamefont {H.}~\bibnamefont {Imai}},\ }\href
  {\doibase 10.1016/s0034-4877(10)00025-x} {\bibfield  {journal} {\bibinfo
  {journal} {Reports on Mathematical Physics}\ }\textbf {\bibinfo {volume}
  {66}},\ \bibinfo {pages} {175} (\bibinfo {year} {2010})}\BibitemShut
  {NoStop}%
\bibitem [{\citenamefont {D'Ariano}\ \emph
  {et~al.}(2020{\natexlab{b}})\citenamefont {D'Ariano}, \citenamefont
  {Perinotti},\ and\ \citenamefont {Tosini}}]{DAriano2020information}%
  \BibitemOpen
  \bibfield  {author} {\bibinfo {author} {\bibfnamefont {G.~M.}\ \bibnamefont
  {D'Ariano}}, \bibinfo {author} {\bibfnamefont {P.}~\bibnamefont {Perinotti}},
  \ and\ \bibinfo {author} {\bibfnamefont {A.}~\bibnamefont {Tosini}},\ }\href
  {\doibase 10.22331/q-2020-11-16-363} {\bibfield  {journal} {\bibinfo
  {journal} {{Quantum}}\ }\textbf {\bibinfo {volume} {4}},\ \bibinfo {pages}
  {363} (\bibinfo {year} {2020}{\natexlab{b}})}\BibitemShut {NoStop}%
\bibitem [{\citenamefont {Popescu}\ and\ \citenamefont
  {Rohrlich}(1994)}]{popescu1994}%
  \BibitemOpen
  \bibfield  {author} {\bibinfo {author} {\bibfnamefont {S.}~\bibnamefont
  {Popescu}}\ and\ \bibinfo {author} {\bibfnamefont {D.}~\bibnamefont
  {Rohrlich}},\ }\href {\doibase 10.1007/BF02058098} {\bibfield  {journal}
  {\bibinfo  {journal} {Foundations of Physics}\ }\textbf {\bibinfo {volume}
  {24}},\ \bibinfo {pages} {379} (\bibinfo {year} {1994})}\BibitemShut
  {NoStop}%
\bibitem [{\citenamefont {D'Ariano}\ and\ \citenamefont
  {Tosini}(2010)}]{d2010testing}%
  \BibitemOpen
  \bibfield  {author} {\bibinfo {author} {\bibfnamefont {G.~M.}\ \bibnamefont
  {D'Ariano}}\ and\ \bibinfo {author} {\bibfnamefont {A.}~\bibnamefont
  {Tosini}},\ }\href@noop {} {\bibfield  {journal} {\bibinfo  {journal}
  {Quantum Information Processing}\ }\textbf {\bibinfo {volume} {9}},\ \bibinfo
  {pages} {95} (\bibinfo {year} {2010})}\BibitemShut {NoStop}%
\bibitem [{\citenamefont {Perinotti}\ \emph {et~al.}(2021)\citenamefont
  {Perinotti}, \citenamefont {Tosini},\ and\ \citenamefont
  {Vaglini}}]{perinotti2021shannon}%
  \BibitemOpen
  \bibfield  {author} {\bibinfo {author} {\bibfnamefont {P.}~\bibnamefont
  {Perinotti}}, \bibinfo {author} {\bibfnamefont {A.}~\bibnamefont {Tosini}}, \
  and\ \bibinfo {author} {\bibfnamefont {L.}~\bibnamefont {Vaglini}},\
  }\href@noop {} {\enquote {\bibinfo {title} {Shannon theory for quantum
  systems and beyond: information compression for fermions},}\ } (\bibinfo
  {year} {2021}),\ \Eprint {http://arxiv.org/abs/2106.04964} {arXiv:2106.04964
  [quant-ph]} \BibitemShut {NoStop}%
\bibitem [{\citenamefont {Perinotti}(2020)}]{Perinotti2020cellularautomatain}%
  \BibitemOpen
  \bibfield  {author} {\bibinfo {author} {\bibfnamefont {P.}~\bibnamefont
  {Perinotti}},\ }\href {\doibase 10.22331/q-2020-07-09-294} {\bibfield
  {journal} {\bibinfo  {journal} {{Quantum}}\ }\textbf {\bibinfo {volume}
  {4}},\ \bibinfo {pages} {294} (\bibinfo {year} {2020})}\BibitemShut {NoStop}%
\bibitem [{\citenamefont {Fuchs}\ and\ \citenamefont {van~de
  Graaf}(1999)}]{761271}%
  \BibitemOpen
  \bibfield  {author} {\bibinfo {author} {\bibfnamefont {C.}~\bibnamefont
  {Fuchs}}\ and\ \bibinfo {author} {\bibfnamefont {J.}~\bibnamefont {van~de
  Graaf}},\ }\href {\doibase 10.1109/18.761271} {\bibfield  {journal} {\bibinfo
   {journal} {IEEE Transactions on Information Theory}\ }\textbf {\bibinfo
  {volume} {45}},\ \bibinfo {pages} {1216} (\bibinfo {year}
  {1999})}\BibitemShut {NoStop}%
\bibitem [{\citenamefont {D'Ariano}\ \emph
  {et~al.}(2020{\natexlab{c}})\citenamefont {D'Ariano}, \citenamefont {Erba},\
  and\ \citenamefont {Perinotti}}]{PhysRevA.102.052216}%
  \BibitemOpen
  \bibfield  {author} {\bibinfo {author} {\bibfnamefont {G.~M.}\ \bibnamefont
  {D'Ariano}}, \bibinfo {author} {\bibfnamefont {M.}~\bibnamefont {Erba}}, \
  and\ \bibinfo {author} {\bibfnamefont {P.}~\bibnamefont {Perinotti}},\ }\href
  {\doibase 10.1103/PhysRevA.102.052216} {\bibfield  {journal} {\bibinfo
  {journal} {Phys. Rev. A}\ }\textbf {\bibinfo {volume} {102}},\ \bibinfo
  {pages} {052216} (\bibinfo {year} {2020}{\natexlab{c}})}\BibitemShut
  {NoStop}%
\bibitem [{\citenamefont {Zhang}(2007)}]{4294175}%
  \BibitemOpen
  \bibfield  {author} {\bibinfo {author} {\bibfnamefont {Z.}~\bibnamefont
  {Zhang}},\ }\href {\doibase 10.1109/TIT.2007.903122} {\bibfield  {journal}
  {\bibinfo  {journal} {IEEE Transactions on Information Theory}\ }\textbf
  {\bibinfo {volume} {53}},\ \bibinfo {pages} {3280} (\bibinfo {year}
  {2007})}\BibitemShut {NoStop}%
\bibitem [{\citenamefont {Fiorini}\ \emph {et~al.}(2015)\citenamefont
  {Fiorini}, \citenamefont {Massar}, \citenamefont {Patra},\ and\ \citenamefont
  {Tiwary}}]{2015JPhA...48b5302F}%
  \BibitemOpen
  \bibfield  {author} {\bibinfo {author} {\bibfnamefont {S.}~\bibnamefont
  {Fiorini}}, \bibinfo {author} {\bibfnamefont {S.}~\bibnamefont {Massar}},
  \bibinfo {author} {\bibfnamefont {M.~K.}\ \bibnamefont {Patra}}, \ and\
  \bibinfo {author} {\bibfnamefont {H.~R.}\ \bibnamefont {Tiwary}},\ }\href
  {\doibase 10.1088/1751-8113/48/2/025302} {\bibfield  {journal} {\bibinfo
  {journal} {J. Phys. A: Math. Theor.}\ }\textbf {\bibinfo {volume} {48}},\
  \bibinfo {pages} {025302} (\bibinfo {year} {2015})}\BibitemShut {NoStop}%
\end{thebibliography}%

\appendix
\section{A simple Lemma}\label{app:trivial_lemma}
\begin{lemma}
Let $\{c\}_{i\in\sX}\subseteq\eff(\rB)$ be an observation test and $\{A\}_{i\in\sX}\subseteq\eff(\rA)$ a collection of effects. If causality holds, then $\sum_{i\in\sX}A_{i}\boxtimes c_{i}\in\eff(\rA\rB)$.
\end{lemma}
\begin{proof}
This is a straightforward consequence of causality. For any $i\in\sX$, there exists 
an observation test $\{\tilde{A}^{(i)}_{j}\}_{j\in\sY_{i}}$ such that  $A_{i}\in\{\tilde{A}^{(i)}_{j}\}_{j\in\sY_{i}}$.
Thus, if we consider the test $\{\tI_{\rB} \boxtimes c_{i}\}_{i\in\sX}$ and the collection of effects
$\{\tilde{A}^{(i)}_{j}\boxtimes c_{i}\}_{(i,j)\in\sX\times\sY}$ we have
\begin{equation}
  \begin{aligned}
    \Qcircuit @C=0.9em @R=.8em @! R {&\ustick{\rA} \qw&\measureD{\tilde{A}^{(i)}_{j}} \\
	& \ustick{\rB} \qw & \measureD{c_{i}}}
  \end{aligned}
  \quad=\quad
  \begin{aligned}
    \Qcircuit @C=0.9em @R=.6em @! R {&\ustick{\rA}\qw&\multigate{1}{\tA_i}&\ustick{\rA} \qw &\gate{\tB_{j}^{i}}&\ustick \rI 
 								 \qw&\qw \\
							& \ustick{\rB} \qw & \ghost{\tA_i} }
  \end{aligned}\ ,
\end{equation}
with $\tA_i\coloneqq\tI_{\rA} \boxtimes c_{i}$ and $\tB_{j}^{i}\coloneqq\tilde{A}^{(i)}_{j}$. Therefore, causality implies that 
$\{\tilde{A}^{(i)}_{j}\boxtimes c_{i}\}_{(i,j)\in\sX\times\sY}$ is an observation test, and 
$\sum_{i\in\sX}A_{i}\boxtimes c_{i}\in\eff(\rA\rB)$, being a coarse graining of effects from the same test.
\end{proof}

\section{Proof of Lemma~\ref{lemma:bound_infocontent}}\label{app:purity_lemmas}
We start by defining the following number
\begin{equation}
\zeta(N,\delta)\coloneqq\sup\{\varepsilon\mid \,E_{N,M,\varepsilon}(\rho)\subseteq E^C_{N,M,\delta}(\rho)\}.
\label{eq:zeta}
\end{equation}
Firstly, we can observe that in the above definition we can safely take the maximum, since the following inclusion holds
\[
E_{N,M,\zeta(N,\delta)}(\rho)\subseteq E^{C}_{N,M,\delta}(\rho).
\]
Indeed, let $(\tE,\tD)\in E_{N,M,\zeta(N,\delta)}$. By definition of $E_{N,M,\zeta(N,\delta)}$ we of have
\[
\sup_{\rC,\{\Psi_i\}}\sum_i\normop{[(\tD\tE-\tI)\boxtimes\tI_{\rC}]\Psi_i}<\zeta(N,\delta)
\]
then there exists $\varepsilon'<\zeta(N,\delta)$ such that $\sup_{\rC,\{\Psi_i\}}\sum_i\normop{[(\tD\tE-\tI)\boxtimes\tI_{\rC}]\Psi_i}<\varepsilon'$. Thus, by definition of $\zeta(N,\delta)$, one has $(\tE,\tD)\in E_{N,M,\varepsilon''}(\rho)$ with $\varepsilon'<\varepsilon''<\zeta(N,\delta)$ and $E_{N,M,\varepsilon''}(\rho)\subseteq E^{C}_{N,M,\delta}(\rho)$. Finally, since $E_{N,M,\varepsilon'}(\rho)\subseteq E_{N,M,\varepsilon''}(\rho)$, we have $(\tE,\tD)\in E^{C}_{N,M,\delta}(\rho)$, and consequently $E_{N,M,\zeta(N,\delta)}(\rho)\subseteq E^{C}_{N,M,\delta}(\rho)$.

This inclusion has another consequence, which is our starting point for proving the lemma. Indeed, by definition one has 
\[
\limsup_{N\rightarrow\infty}R_{\zeta(N,\delta),N}(\rho)\geq R_{\delta}^{C}(\rho)
\]
We now have the following two possibilities
\begin{enumerate}
\item $\exists\delta_{0}>0$ such that, $\forall0<\delta<\delta_0$, $\liminf_{N\rightarrow\infty}\zeta(N,\delta)=0$;
\item $\forall\delta>0$ one has $\liminf_{N\rightarrow\infty}\zeta(N,\delta)=:\overline{\zeta}(\delta)>0$.
\end{enumerate}
Let us start analysing case 2. In this case, by definition of limit inferior, one has
\begin{align*}
&\forall\delta,\gamma>0\\
&\left\{
\begin{aligned}
&\exists N_0,\quad\forall N\geq N_0,\quad \zeta(N,\delta)>\bar\zeta(\delta)-\gamma,\\
&\forall N_0,\quad\exists N\geq N_0,\quad \zeta(N,\delta)<\bar\zeta(\delta)+\gamma.
\end{aligned}
\right.
\end{align*}
This implies that for every $\delta>0$ and every positive $\gamma$, for suitably large $N$ it is $R_{\overline{\zeta}(\delta)-\gamma,N}(\rho)\geq R_{\zeta(N,\delta),N}(\rho)$, and consequently, for suitably large $N$ it is $R_{\overline{\zeta}(\delta)/2,N}(\rho)\geq R_{\zeta(N,\delta),N}(\rho)$. In turn, this implies 
\[
R_{\overline{\zeta}(\delta)/2}(\rho)\geq\limsup_{N\rightarrow\infty}R_{\zeta(N,\delta),N}(\rho)\geq R_{\delta}^{C}(\rho),
\]
and finally, being $\overline{\zeta}(\delta)$ increasing as a function of $\delta$, taking the limit for $\delta\rightarrow 0$ one has some value $\varepsilon\geq 0$ such that
\[
I(\rho)\geq\lim_{\zeta\rightarrow\varepsilon}R_{\zeta}(\rho)=\lim_{\delta\rightarrow 0}R_{\overline{\zeta}(\delta)}(\rho)\geq I^{C}(\rho).
\]

We now turn to case 1, and show that this is not possible. The hypotheses imply indeed that there exists $\delta_0>0$ such that 
$\liminf_{N\rightarrow\infty}\zeta(N,\delta_0)=0$, and the same is then true of every $0<\delta\leq\delta_0$. This means that for 
every $\gamma>0$ and every $N_0$ there exists $N\geq N_0$ such that $\zeta(N,\delta)<\gamma$ for all $0<\delta\leq\delta_0$.
By definition, this means that for every $\gamma$ there exists a scheme $(\tE,\tD)\in E_{N,M,\gamma}(\rho)$ such that 
$(\tE,\tD)\notin E_{N,M,\delta}^{C}(\rho)$. More explicitly
\begin{align*}
&\sup_{\rC,\{\psi_i\}}\,\sum_i\normop{[(\tD\tE-\tI)\boxtimes\tI_\rC]\psi_i}<\gamma,\\
&\sup_{\rC,\{\psi_i\},\{a_j\}}\,L^{-1}|I(\rX:\rY)-I(\rX:\tilde \rY)|>\delta,
\end{align*}
where $L$ has been introduced in def.~\ref{def:classinfoc}. First of all we remark that if $m=1$ or $n=1$, then $H(\rX)=0$ or
$H(\rY)=H(\tilde{\rY})=0$, respectively, and thus $I(\rX:\rY)=I(\rX:\tilde{\rY})=0$, since $I(\rA:\rB)\leq\min\{H(\rA),H(\rB)\}$. The minimum relevant value of
$L$ is thus $\log_2 3$. Now according to theorem 2 in~\cite{4294175}, for $\norma{\mathbf p-\mathbf q}<\gamma<1-1/mn$ one has
\begin{align*}
&L^{-1}|I(\rX:\rY)-I(\rX':\rY')| \\
\leq &3\gamma+3L^{-1}H_2(\gamma) \\
\leq & 3\gamma+\frac3{\log_23}H_2(\gamma)
\end{align*}
where $\rX,\rY$ and $\rX',\rY'$ are distributed according to $p_{i,j}$ and $q_{i,j}$, respectively. We can then conclude that for every $\gamma>0$ one has
\[
\delta<3\gamma+\frac3{\log_23}H(\gamma).
\]
However, our hypotheses imply that the latter condition must hold for some $\delta>0$, which is absurd.
\section{Proof of Lemma~\ref{lemma:posbound_classinfocontent}}\label{app:proofincubo}
Let us take $\delta>0$, and consider $(\tE,\tD)\in E^C_{N,M,\delta}(\rho)$. Let us consider first a single use of the source associated
with $\rho$ corresponding to the decomposition $\{\Psi_i\}$, and let $\{ a_j\}$ be the observation test such that $I(\rX_0:\rY_0)$ is maximum, where $\rX_0$ is the classical variable corresponding to the outcome $i$ of the preparation test, and $\rY_0$ that of the 
observation test. Notice that by Krein-Millman's theorem and Caratheodory's theorem one can always find the supremum of 
mutual information considering atomic decompositions and observation-tests with a bounded number of elements, and thus the
optimisation problem has a compact domain. Let now $\{\Psi_{\bvec i}\}$ be the decomposition of $\rho^{\boxtimes N}$ defined
by
\[
\Psi_{\bvec i}\coloneqq\Psi_{i_1}\boxtimes\Psi_{i_2}\boxtimes\dots\boxtimes\Psi_{i_N},
\]
and $\Psi_i$ be the decomposition that maximises $I(\rX_0:\rY_0)$, with $m_0$ outcomes.
Let now $\{ b_j \}$ be the observation test on N copies of the system that maximizes $I(\rX:\rY)$ where $\rX$ is the i.i.d.~classical 
variable given by the preparation event $\bvec i$ and $\rY$ by the outcome $j$. Since $\{b_j\}$ maximizes the mutual information
it is clear that the test $\{(b_j|\tD\tE\}$ will provide a mutual information $I(\rX:\tilde{\rY})$ no larger than $I(\rX:\rY)$. Thus we can write
\begin{align*}
\delta&>\frac{I(\rX:\rY)-I(\rX:\tilde{\rY})}{\log_2 m_0^ND(N)-1}\\
&\geq \frac{I(\rX:\rY)-I(\rX:\tilde{\rY})}{N\log_2 m_0D_0+\log_2k},
\end{align*}
where in the first bound we used the fact that the number of outcomes for the observation test maximising the mutual information
does not exceed the dimension of the space of effects $D(N)$, while in the second bound we used the hypothesis that there exist $k,D_0$ such that $D(N)\leq k D_0^N$. Now, by definition of $I(\rX:\rY)$ we have $I(\rX:\rY)\geq NI(\rX_0 : \rY_0)$, while by the result of theorem 2 in~\cite{2015JPhA...48b5302F} we have 
\[
I(X:\tilde{Y})\leq \log_2 D(M)\leq log_2 k' + M\log_2 D_1,
\]
where we think of the scheme given by the decomposition $\{ \tE\state{\Psi_{\bvec i}} \}$ and the observation test given by
$\{ \effect{b_j}\tD \}$, involving $M$ obits. We can then write the following inequality
\[
\delta > \frac{NI(X_0:Y_0) -M\log_2 D_1 - \log_2 k'}{N\log_2 m_0 D_0 + \log_2 k},
\]
and consequently 
\begin{align*}
& \frac{M}{N}\frac{\log_2 D_1 + \log_2 k' / M}{\log_2 m_0 D_0 +\log_2 k/N}  +\delta\\
& >\frac{I(X_0:Y_0)}{\log_2 m_0D_0 + \log_2 k/N}.
\end{align*}
In particular, if the scheme $(\tE,\tD)$ has the minimum M for fixed $N,\delta$ we can then conclude that
\begin{align*}
R_{\delta,N}^C&\frac{\log_2 D_1}{\log_2 m_0 D_0 + \frac{\log_2 k}{N}} +\delta\\
&+ \frac{1}{N}\frac{\log_2 k'}{\log_2 m_0 D_0 + \frac{\log_2 k}{N}}\\
>&\frac{I(X_0:Y_0)}{\log_2 m_0 D_0 + \frac{\log_2 k}{N}}.
\end{align*}
Taking te limit superior for $N\to\infty$ on both sides we have
\[
R_{\delta}^C\frac{\log_2 D_1}{\log_2 m_0D_0}+\delta \geq \frac{I (X_0:Y_0)}{\log_2 m_0D_0},
\]
and finally, in the limit $\delta\rightarrow 0$ we obtain
\[
I^C(\rho)\frac{\log_2 D_1}{\log_2 m_0D_0}\geq\frac{I(X_0:Y_0)}{\log_2 m_0D_0},
\]
namely
\[
I^C(\rho)\geq\frac{I(X_0:Y_0)}{\log_2 D_1}.
\]
For a mixed state, $I(X_0:Y_0)>0$ and this implies the thesis.

\end{document}